\pdfoutput=1

\documentclass[acmsmall,screen,nonacm]{acmart}

\usepackage{enumitem}
\usepackage[noorphans,vskip=0ex]{quoting}
\PassOptionsToPackage{vlined, boxruled}{algorithm2e}
\usepackage{algorithm2e}

\usepackage{mleftright}
\usepackage{relsize}

\usepackage{amsmath}
\usepackage{amsthm}
\usepackage{tabularx}
\newcolumntype{Y}{>{\centering\arraybackslash}X}

\usepackage{array}
\usepackage{soul}
\usepackage{mdframed}
\usepackage{xparse}
\usepackage[capitalise, noabbrev]{cleveref}
\usepackage[caption=false]{subfig}

\usepackage{tcolorbox}
\usepackage{tikz,lipsum}
\tcbuselibrary{skins,breakable}
\SetKwInOut{Input}{Input}
\SetKwProg{Initialization}{Initialization}{}{}
\SetKwProg{Fn}{Function}{}{end}
\SetKwProg{EFn}{Event Function}{}{end}
\SetKwFunction{CreateInstance}{CreateInstance}
\SetKw{And}{and}
\SetKw{Or}{or}
\SetKw{Break}{break}
\SetKw{Continue}{continue}
\SetKw{Goto}{goto}
\SetKwFor{ContInc}{continually increase}{until:}{}

\SetCommentSty{mycommfont}

\SetFuncSty{myfuncsty}

\LinesNumbered
\RestyleAlgo{boxruled}
\DontPrintSemicolon

\makeatletter
\newcommand{\algorithmfootnote}[2][\footnotesize]{%
    \let\old@algocf@finish\@algocf@finish
    \def\@algocf@finish{\old@algocf@finish
    \leavevmode\rlap{\begin{minipage}{\linewidth}
                         #1#2
    \end{minipage}}%
    }%
}
\makeatother

\definecolor{darkred}{RGB}{200,0,0}

\newcommand{\eat}[1]{}

\newcommand{\opt}{\text{\rm OPT}}

\newcommand{\alg}{\text{\rm ALG}}

\newcommand{\pr}[1]{\mleft(#1\mright)}

\newcommand{\pc}[1]{\mleft\{#1\mright\}}

\newcommand{\ceil}[1]{\mleft\lceil#1\mright\rceil}

\newcommand{\cset}[2]{\pc{#1\middle|#2}}

\newtheorem{observation}{Observation}

\setlist{itemsep=0pt}

\SetKwFunction{Open}{Open}
\SetKwFunction{Invest}{Invest}
\SetKwFunction{Explore}{Explore}

\SetKwFunction{SetColor}{SetColor}

\SetKwFunction{ComputePotentials}{ComputePotentials}
\SetKwFunction{InitService}{InitService}

\SetKwFunction{Charge}{Charge}

\NewDocumentCommand{\tree}{o}{
    \IfNoValueTF{#1}{T}{T_{#1}}
}

\NewDocumentCommand{\loc}{m}{x(\req)}

\NewDocumentCommand{\distsgn}{}{\delta}
\NewDocumentCommand{\dist}{O{}mm}{\distsgn_{#1}\pr{#2, #3}}

\NewDocumentCommand{\req}{}{q}
\NewDocumentCommand{\reqs}{}{Q}
\NewDocumentCommand{\creqs}{}{Q^{\star}}
\NewDocumentCommand{\creq}{}{\req^{\star}}

\NewDocumentCommand{\pen}{}{\pi}

\NewDocumentCommand{\ar}{}{\Delta}
\NewDocumentCommand{\npoint}{}{n}

\NewDocumentCommand{\level}{m}{\ell_{#1}}
\NewDocumentCommand{\alevel}{m}{\overline{\ell}_{#1}}
\NewDocumentCommand{\stime}{m}{t_{#1}}
\NewDocumentCommand{\ftime}{m}{\tau_{#1}}
\NewDocumentCommand{\ptime}{m}{\sigma_{#1}}

\NewDocumentCommand{\ms}{}{G}

\NewDocumentCommand{\serv}{}{\lambda}
\NewDocumentCommand{\servs}{}{\Lambda}
\NewDocumentCommand{\pservs}{}{\servs^{\mathrm{p}}}
\NewDocumentCommand{\pservsf}{}{\servs^{\mathrm{pf}}}
\NewDocumentCommand{\pservss}{}{\servs^{\mathrm{ps}}}
\NewDocumentCommand{\cservs}{}{\servs^{\mathrm{c}}}
\NewDocumentCommand{\ball}{mm}{\ballsgn\pr{#1,#2}}
\NewDocumentCommand{\perball}{mmm}{\ballsgn^{#1}\pr{#2,#3}}
\NewDocumentCommand{\nreqs}{}{m}

\NewDocumentCommand{\ballsgn}{}{B}
\NewDocumentCommand{\elig}{}{E}

\NewDocumentCommand{\cyl}{}{\gamma}
\NewDocumentCommand{\percyl}{m}{\cyl^{#1}}
\NewDocumentCommand{\cyls}{}{\Gamma}
\NewDocumentCommand{\percyls}{}{\mathscr{H}}
\NewDocumentCommand{\pdchg}{m}{D^{*}_{\mathrm{p}, #1}}
\NewDocumentCommand{\prp}{m}{\pr{#1}^{+}}

\NewDocumentCommand{\cdchg}{m}{D^{*}_{\mathrm{c},#1}}
\NewDocumentCommand{\cpen}{m}{\pen_{\mathrm{c}, #1}}
\NewDocumentCommand{\pcyls}{}{\cyls^{\mathrm{p}}}
\NewDocumentCommand{\pcylsf}{}{\cyls^{\mathrm{pf}}}
\NewDocumentCommand{\pcylss}{}{\cyls^{\mathrm{ps}}}
\NewDocumentCommand{\ccyls}{}{\cyls^{\mathrm{c}}}
\NewDocumentCommand{\pcyl}{m}{\cyl_{\mathrm{p}}(#1)}
\NewDocumentCommand{\ccyl}{m}{\cyl_{\mathrm{c}}(#1)}
\NewDocumentCommand{\ppercyls}{}{\percyls^{\mathrm{pf}}}

\NewDocumentCommand{\I}{mm}{(#1, #2]}
\NewDocumentCommand{\rlt}{m}{r_{#1}}
\NewDocumentCommand{\dlt}{m}{d_{#1}}
\NewDocumentCommand{\ylt}{mm}{d_{#1}(#2)}
\NewDocumentCommand{\ctr}{}{h}
\NewDocumentCommand{\rylt}{mm}{y_{#1}(#2)}
\NewDocumentCommand{\trd}{}{Y}
\NewDocumentCommand{\ivl}{}{I}
\NewDocumentCommand{\civl}{m}{\ivl_{\mathrm{c}}(#1)}
\NewDocumentCommand{\pivl}{m}{\ivl_{\mathrm{p}}(#1)}
\NewDocumentCommand{\ccap}{mm}{\cost{#1 \cap #2}}
\NewDocumentCommand{\levsep}{}{b}

\NewDocumentCommand{\stree}{}{\mathrm{ST}}
\NewDocumentCommand{\pcstree}{}{\mathrm{PCST}}

\SetKwFunction{UponRequest}{UponRequest}
\SetKwFunction{UponDeadline}{UponDeadline}
\SetKwFunction{UponCritical}{UponCritical}
\SetKwFunction{ForwardTime}{ForwardTime}
\newcommand{\lIfElse}[3]{\lIf{#1}{#2 \textbf{else}~#3}}
\NewDocumentCommand{\cost}{m}{
    c\pr{#1}
}

\NewDocumentCommand{\coston}{mmo}{
    \IfNoValueTF{#3}{
        #1\pr{#2}
    }{
        #1\pr{#2|#3}
    }
}

\title{Improved and Deterministic Online Service with Deadlines or Delay}

\author{Noam Touitou}
\orcid{0000-0002-5720-4114}
\affiliation{%
  \institution{Amazon}
  \city{Tel Aviv}
  \country{Israel}
}
\email{noamtwx@gmail.com}



\date{}
\begin{document}

    \begin{abstract}
        We consider the problem of online service with delay on a general metric space, first presented by Azar, Ganesh, Ge and Panigrahi (STOC 2017).
The best known randomized algorithm for this problem, by Azar and Touitou (FOCS 2019), is $O(\log^2 \npoint)$-competitive, where $\npoint$ is the number of points in the metric space.
This is also the best known result for the special case of online service with deadlines, which is of independent interest.

In this paper, we present $O(\log \npoint)$-competitive \emph{deterministic} algorithms for online service with deadlines or delay, improving upon the results from FOCS 2019.
Furthermore, our algorithms are the first deterministic algorithms for online service with deadlines or delay which apply to general metric spaces and have sub-polynomial competitiveness.
    \end{abstract}

    \maketitle

%

    \section{Introduction}
    \label{sec:Intro}
    In online service with deadlines/delay, a server exists on a metric space of $\npoint$ points.
Requests arrive over time on points in the metric space, demanding service by the algorithm.
The algorithm can serve requests by moving the server to their location, incurring a cost which is the distance traveled by the server on the metric space.
In \emph{online service with deadlines}, each request has an associated deadline by which it must be served.
In \emph{online service with delay}, a more general problem, the deadline is replaced with delay costs which accrue while the request is pending.
Specifically, each request has an associated, non-decreasing delay function, such that the total delay cost incurred by a pending request until time $t$ is the value of its delay function at $t$.

Online service with delay was first introduced by Azar et al.~\cite{DBLP:conf/stoc/AzarGGP17}, who gave an $O(\log^4 \npoint)$ randomized algorithm for the problem, based on randomized embedding of the metric space into a tree (specifically, a weighted hierarchically well-separated tree), then solving the problem on the resulting tree.
In~\cite{DBLP:conf/focs/AzarT19}, this was improved to $O(\log^2 \npoint)$-competitiveness, through an improved algorithm for online service on a tree.
The result of~\cite{DBLP:conf/focs/AzarT19} remains the best known randomized result for this problem.

Without randomization, much less is known about this problem.
There is no known deterministic algorithm (of competitiveness less than polynomial) which applies to general metric spaces.
For \emph{specific} metric spaces, some results are known.
When the metric space is uniform (or weighted uniform), the work of Azar et al.~\cite{DBLP:conf/stoc/AzarGGP17} implies a constant-competitive deterministic algorithm.
When the metric space is a line, Bienkowski et al.~\cite{DBLP:conf/sirocco/BienkowskiKS18} presented an $O(\log \ar)$-competitive deterministic algorithm; here, $\ar$ is the aspect ratio of the metric space, or the ratio between the largest and smallest pairwise distances (for a line, note that $\ar \ge \npoint$).

\subsection{Our Results}
We consider online service with deadlines/delay on a metric space of $\npoint$ points, and present the following results.

\begin{enumerate}
    \item An $O(\log \npoint)$-competitive, deterministic algorithm for online service with deadlines that runs in polynomial time. (\Cref{sec:OSD}.)
    \item An $O(\log \npoint)$-competitive, deterministic algorithm for online service with delay that runs in polynomial time. (\Cref{sec:OSY}.)
\end{enumerate}

Both results improve upon the best known \emph{randomized} algorithm for service with deadlines/delay, which is the randomized $O(\log^2 \npoint)$-competitive algorithm of~\cite{DBLP:conf/focs/AzarT19}.
Moreover, these are the first deterministic algorithms of sub-polynomial competitiveness for online service with deadlines/delay on general metric spaces.
Note that while the result for deadlines is implied by the result for delay, we chose to present it independently.
This is both for ease of presentation and since the deadline case is of independent interest.

In fact, we show that our algorithms achieve a stronger result: they are $O(\log \min\pc{\npoint, \nreqs})$-competitive, where $\nreqs$ is the number of requests in the input.
Note that previous algorithms had no guarantee in terms of the number of requests.
Specifically, previous algorithms were based on randomized tree embedding, and thus lose $\Theta(\log \npoint)$ in competitiveness even when $\nreqs$ is constant.
We discuss this result in \cref{sec:RR}.

\subsection{Our Techniques}

\textbf{Online service with deadlines.}.
The algorithm for online service with deadlines employs the main concept used in~\cite{DBLP:conf/focs/AzarT20} for network design problems with deadlines, which is to assign levels to requests that increase over time, such that high-level requests are only served in high-cost services.
Services also have levels, determined by the level of request whose deadline has been reached.
The budget of a service is exponential in its level, and a service of level $\ell$ only serves requests of level at most $\ell$.

However, in network design the cost of serving a request is fixed at all times, while in online service this cost depends on the current location of the server.
Thus, while levels are maintained for each request, the participation of a request in a service depends both on its level and on its distance from the server at the time of service; the resulting parameter is called the \emph{adjusted level} of the request.
As a service of level $\ell$ restricts itself to requests with adjusted level at most $\ell$, each service is confined to some ``service ball'' centered at the server's location.

In addition, online service calls for a more aggressive raising of levels.
In particular, the following properties are crucial to the analysis:
\begin{enumerate}
    \item Upon the deadline of a request, a service is started with level much larger than that of the triggering request.
    \item Upon the end of a service, requests in its service ball are upgraded to a \emph{higher} level than that of the service.
\end{enumerate}
Compare this to the network design framework of~\cite{DBLP:conf/focs/AzarT20}, in which the constant difference in levels between service and triggering request can be any positive number without breaking the analysis (and is chosen to be $1$).
In addition, for network design the level of an eligible request only increased to the level of the service.

Finally, the server itself must occasionally move to more opportune locations.
In our algorithm, this depends on the triggering request: if its adjusted level is dictated by its distance from the server, we say that the service is \emph{primary}.
In this case, the server is moved to the request at the end of the service; otherwise, the server returns to its initial position.

\textbf{Analysis.}
The optimal solution for online service is harder to characterize than in network design.
Optimal services in network design can be charged independently (where the requests served are an ``intersecting set''), while the tour of the optimal server in online service is not easily partitioned.
This calls for a novel type of analysis that we introduce.

In our analysis, we construct space-time \emph{cylinders}, where each cylinder is associated with a shape in the metric space (e.g., a ball) and a time interval.
We then show two properties: first, that the optimal solution incurs enough cost inside each cylinder; second, that the cylinders are disjoint (either temporally or spatially).
This yields a lower bound on the cost of the optimal solution; as each cylinder is associated with a service in the algorithm, this connects the cost of the algorithm with that of the optimal solution.
The aggressive upgrading of requests and services is dictated by this analysis: upgrading requests forces $\opt$ to incur high cost inside each cylinder, while upgrading services implies disjointness of the cylinders.
We first use cylinders in a simple way, such that the cylinders' associated shapes are balls in the metric space.
Then, to show our final result, we perforate those balls by removing balls of much-smaller radius; we then show that the charge to the optimum is maintained, while achieving a greater degree of cylinder disjointness (and thus a better competitive ratio).

\textbf{Online service with delay.}
The algorithm for online service with delay is similar to the algorithm for deadlines.
In deadlines, a service is started when the deadline of a request expires; in delay, a service is started when the total delay for a set of pending requests becomes large.
(In fact, we consider the \emph{residual} delay for this condition, which is the amount of delay that exceeds investment from past services.)
Specifically, when requests of adjusted level at most $\ell$ gather a total delay of $2^\ell$, we say that level $\ell$ has become critical and trigger a service.
This service uses a \emph{prize-collecting} algorithm for Steiner tree to choose whether to serve requests or invest in them, offsetting their future delay.
This use of a prize-collecting approximation algorithm is similar to that in~\cite{DBLP:conf/focs/AzarT20}.

The salient difference between the deadline and delay case lies in moving the server upon a primary service (where ``primary'' is defined somewhat similarly to deadlines).
Considering deadlines as a specific case of delay, an expired deadline is equivalent to infinite delay incurred in a very concentrated neighborhood (i.e., a single point).
Analogously, for the case of delay, in a primary service we attempt to identify a small-radius ball within which a constant fraction of the residual delay exists.
If such a ball is found, the server would move to its center at the end of the service.
Otherwise, the delay is well-spread, and the server would remain stationary.
The intuition here is that when the delay is well-spread, the optimal solution must also make significant movements to avoid incurring large delay.

\subsection{Related Work}
\textbf{Multiple servers.}
Online service with delay has also been considered when the algorithm has $k>1$ servers.
In the first paper of Azar et al.~\cite{DBLP:conf/stoc/AzarGGP17}, an $O(k \cdot \textrm{poly}\log(\npoint))$-competitive randomized algorithm was given for this problem.
As the algorithm only used randomization in the initial embedding stage, its dependence on $k$ is linear (as online service with delay is a generalization of the $k$-server problem, which has an $\Omega(k)$-competitiveness lower bound for deterministic algorithms).
For uniform metric spaces, a better use of randomization was done in~\cite{DBLP:journals/siamcomp/GuptaKP22}, achieving an $O(\log \npoint \log k)$-competitive algorithm.
For online service with deadlines on a general metric space, an $O(\textrm{poly}\log (\ar\npoint))$-competitive randomized algorithm was presented by Gupta et al.~\cite{DBLP:conf/focs/00010P21}.

\textbf{Network design with deadlines/delay.}
A set of problems related to online service is network design with deadlines/delay.
In such problems, connectivity requests with deadlines or delay arrive over time, and must be served by transmitting a subgraph that provides the desired connectivity.
A notable example is Steiner tree, in which requests demand connecting a terminal to some root node.
It can be seen that Steiner tree with deadlines/delay is a special case of online service with deadlines/delay: the reduction involves forcing the server to remain at the root node through a stream of requests with immediate deadline (or very high delay cost).
The special case of Steiner tree in which the metric space is itself a tree is called multilevel aggregation, and has received much attention~\cite{DBLP:conf/esa/BienkowskiBBCDF16,DBLP:conf/soda/BuchbinderFNT17,DBLP:conf/focs/AzarT19}; this is true also of its special cases, TCP acknowledgement (e.g., \cite{TCPAck_DBLP:conf/stoc/DoolyGS98,TCPAck_DBLP:journals/algorithmica/KarlinKR03,TCPAck_DBLP:conf/esa/BuchbinderJN07}) and joint replenishment (\cite{JointRep_DBLP:conf/soda/BuchbinderKLMS08,JointRep_DBLP:journals/algorithmica/BritoKV12,JointRep_DBLP:conf/soda/BienkowskiBCJNS14}).
Multilevel aggregation itself has also yielded algorithms for online service with delay (see~\cite{DBLP:conf/focs/AzarT19}).
A general framework for network design problems was introduced in~\cite{DBLP:conf/focs/AzarT20}; some techniques introduced in~\cite{DBLP:conf/focs/AzarT20} are used in this paper for online service.
These techniques were also used by Chen et al.~\cite{DBLP:conf/icalp/ChenKU22} for a generalization of joint replenishment.

\textbf{Classic online $k$-server.}
In $k$-server, the classic variant of online service with deadlines/delay, requests arrive over a sequence rather than over time to be served by one of $k$ servers.
(Here, unlike in service with deadlines/delay, the case of a single server is trivial.)
Deterministically, the best competitiveness bound for this problem is $\Theta(k)$~\cite{DBLP:journals/jal/ManasseMS90,DBLP:journals/jacm/KoutsoupiasP95}, where determining the exact constant is an open problem.
With randomization, poly-logarithmic competitive ratios have been achieved relatively recently~\cite{DBLP:conf/focs/BansalBMN11,DBLP:conf/stoc/BubeckCLLM18}.

    \section{Preliminaries}
    \label{sec:Prelim}

In online service with deadlines/delay, we are given a metric space of $\npoint$ points.
We represent this metric space as a weighted, simple graph $\ms$ of $\npoint$ nodes, such that the distance $\dist{u}{v}$ between two points in the metric space is the weight of the shortest path between the nodes in the graph.
Each request $\req$ in the input request set $\reqs$ arrives at time $\rlt{\req}$; slightly abusing notation, we also use $\req$ to denote the point in $\ms$ on which the request exists.
A server exists in the metric space, such that moving the server to a pending request $\req$ serves the request (the server movements are immediate, and do not require time).

In the deadline case, each request $\req \in \reqs$ has deadline $\dlt{\req}$, and must be served in the interval $\I{\rlt{\req}}{\dlt{\req}}$; we assume WLOG that the deadlines of all requests are distinct (this can be enforced by arbitrary tie breaking by the algorithm).
The goal is to minimize the total movement of the server during the course of the algorithm, while still serving all requests by their deadline.

In the delay case, each request $\req \in \reqs$ has a nondecreasing delay function $\ylt{\req}{t}$, defined for every time $t\ge \rlt{\req}$, such that the total delay cost that pending request $\req$ accrues by time $t$ is $\ylt{\req}{t}$.

Without loss of generality, we assume that $\ylt{\req}{\rlt{\req}} = 0$, and that delay rises continuously.
(Indeed, the former assumption translates to an additive constant to every solution, while the latter can again be enforced by the algorithm.)
For ease of presentation, and in keeping with some previous work, we also assume that the delay of every request tends to infinity as time advances; that is, that every request must be served eventually.
(We remark that our algorithm can also be seen to work without this assumption.)

For every number $x$, we define $\prp{x} := \max\pc{x,0}$.
Given a point $v \in \ms$ and a radius $r$, we define $\ball{v}{r}$ to be the set of nodes $u \in \ms$ such that $\dist{v}{u} \le r$.

    \section{Online Service with Deadlines}
    \label{sec:OSD}
    In this section, we consider online service with deadlines, and prove the following theorem.

\begin{theorem}
    \label{thm:OSD_Competitiveness}
    There exists an $O(\log \npoint)$-competitive deterministic algorithm for online service with deadlines.
\end{theorem}

\subsection{The Algorithm}

\textbf{Steiner tree.}
Our algorithm contains a component which produces an approximate solution to the (offline) Steiner tree problem.
In this problem, one is given a set of terminal nodes in a graph, and must output a minimum-cost subtree spanning those terminals.
A classic result for approximating offline Steiner tree~\cite{DBLP:journals/acta/KouMB81} shows that there exists a 2-approximation for this problem; we denote this approximation algorithm by $\stree$, such that $\stree(U)$ denotes the output of the algorithm on the graph $\ms$ given the set of terminals $U$.
Slightly abusing notation, we also use $\stree(U)$ to denote the \emph{cost} of the approximate solution.
Similarly, we use $\stree^*(U)$ to denote some optimal solution for Steiner tree on terminals $U$ (or its cost).

\textbf{Algorithm's description.}
We now describe the behavior of the algorithm for online service with deadlines.
The algorithm is divided into \emph{services}, which are instantaneous events in which the algorithm decides to move its server to serve requests.
For every pending request $\req$, the algorithm maintains a level $\level{\req}$, which limits the set of services for which the request can be eligible (initially, $\level{\req}=-\infty$).
In addition, with $a$ the current location of the server, we define the \emph{adjusted level} of a request $\req$ to be
\[
    \alevel{\req} := \max\pc{\level{\req},\ceil{ \log \dist{a}{\req}}}
\]

Upon deadline of request $\req$, the algorithm starts a new service $\serv$.
The service $\serv$ also has a level $\level{\serv}$, which is larger by a constant from the adjusted level of the triggering request $\req$.
The level of $\serv$ determines which pending requests are considered for service by $\serv$; specifically, a request $\req'$ is eligible for service only if $\alevel{\req'} \le \level{\serv}$.
This means that $\serv$ restricts itself to requests that are both of level at most $\level{\serv}$, and are within the ball $\ball{a}{2^{\level{\serv}}}$, where $a$ is the current location of the server.

Once the eligible requests have been identified, the algorithm attempts to solve them by order of increasing deadlines, subject to a budget of $\Theta(2^{\level{\serv}})$.
This makes use of a Steiner tree approximation component, to design an efficient path through the chosen requests.
The algorithm then traverses this path, serving the chosen requests and finishing at its starting position $a$.
For the remaining, unserved eligible requests, their level is raised to \emph{above} the level of the service; specifically, to level $\level{\serv}+1$.

Finally, note that for the triggering request $\req$,  $\alevel{\req}$ is dictated by either $\level{\req}$ or $\dist{\req}{a}$.
If it is dictated by the latter, the service is called a \emph{primary} service, and the service would move the server from $a$ to $\req$.
Otherwise, the server would remain at $a$ at the end of the service.
The pseudocode description of the algorithm is given in \cref{alg:OSD}.

\begin{algorithm}
    \caption{Online Service with Deadlines}
    \label{alg:OSD}
    \EFn{\UponRequest{$\req$}}{
        set $\level{\req} \gets -\infty$.
    }
    \EFn{\UponDeadline{$\req$}}{
        start a new service, denoted by $\serv$.

        let $a$ be the current location of the server.


        \lIfElse{$\alevel{\req} \neq  \level{\req}$}{say $\serv$ is primary}{$\serv$ is not primary.}\label{line:OSD_SetPrimary}

        set service level $\level{\serv} \gets \alevel{\req} + 3$.

        let $\reqs'$ be the set of currently pending requests.

        let $\elig_{\serv} \gets \cset{\req' \in \reqs'}{\alevel{\req'} \le \level{\serv}}$.
        \label{line:OSD_DefineEligible}

        let $\reqs_{\serv} \gets \pc{\req}, S\gets \emptyset$.

        \For{$\req' \in \elig_{\serv}$ by order of increasing deadline}{
            add $\req'$ to $\reqs_{\serv}$.

            let $S \gets \stree(\reqs_{\serv})$.

            \If{$\cost{S} \ge 4\cdot 2^{\level{\serv}}$}{
                \Break from the loop. \label{line:OSD_TreeReachedBudget}
            }
        }

        perform DFS tour of $S$, serving $\reqs_{\serv}$.

        \ForEach{$\req' \in \elig_{\serv} \setminus \reqs_{\serv}$}{
            set $\level{\req'} \gets \level{\serv} + 1$.\label{line:OSD_LevelUpgrade}
        }

        \lIf{$\serv$ is primary}{move the server to $\req$.}
    }
\end{algorithm}

\subsection{Analysis}

Our goal now is to analyze \cref{alg:OSD} and prove \cref{thm:OSD_Competitiveness}.
Recall that we denote by $\ar$ the aspect ratio of the metric space, i.e., the ratio between the largest and smallest pairwise distances.
For ease of exposition, we first prove the following weaker theorem;

\begin{theorem}
    [weaker version of \cref{thm:OSD_Competitiveness}]
    \label{thm:OSD_WeakCompetitiveness}
    There exists an $O(\log (\npoint\ar))$-competitive deterministic algorithm for online service with deadlines.
\end{theorem}

After proving \cref{thm:OSD_WeakCompetitiveness}, we show how to strengthen some components in the analysis to obtain \cref{thm:OSD_Competitiveness}.

\subsubsection*{Basic Definitions and Properties.}

We denote by $\servs$ the set of services performed by the algorithm.

\begin{definition}[basic service definitions]
    Let $\serv \in \servs$ be a service.
    We define:
    \begin{itemize}
        \item The \emph{triggering request} of $\serv$, denoted $\creq_{\serv}$, to be the request whose deadline started $\serv$.
        \item The location $a_{\serv}$ to be the initial location of the server when $\serv$ is triggered.
        \item The service time $\stime{\serv}:=\dlt{\creq_{\serv}}$.
        \item The request set $\reqs_{\serv}$ to be the requests served by $\serv$ (i.e., the final value of the variable of that name in $\UponDeadline$).
        \item The request set $\elig_{\serv}$, as defined in \cref{line:OSD_DefineEligible}; these requests are called \emph{eligible for $\serv$}.
        \item The \emph{forwarding time} of $\serv$, denoted $\ftime{\serv}$, to be the maximum deadline of a request in $\reqs_{\serv}$.
        \item The \emph{cost} of $\serv$, denoted by $\cost{\serv}$, to be the total cost of moving the server in $\serv$.
        For a set of services $\servs'$, we define $\cost{\servs'} := \sum_{\serv \in \servs'} \cost{\serv}$.
    \end{itemize}
\end{definition}

We now define two subsets of services which are of particular focus: \emph{primary} services and \emph{certified} services.
(Note that a service can belong to both subsets.)
We later show that bounding the costs of these two subsets is enough to bound the total cost of the algorithm.

\begin{definition}[primary services]
    A service $\serv \in \servs$ is called \emph{primary} if it is set to be primary in \cref{line:OSD_SetPrimary}; that is, if $\alevel{\creq_{\serv}} \neq \level{\creq_{\serv}}$ at $\stime{\serv}$.
    We denote by $\pservs \subseteq \servs$ the set of primary services in the algorithm.
\end{definition}

\begin{definition}[witness requests and certified services]
    At a certain point in time, a request $\req$ is called a \emph{witness} for a service $\serv$ if its level $\level{\req}$ was last modified by $\serv$ (at \cref{line:OSD_LevelUpgrade}).

    Note that the triggering request of a non-primary service $\serv'$ is always a witness for an earlier service $\serv$; we say that $\serv'$ \emph{certifies} $\serv$, and call $\serv$ a \emph{certified} service.
    We denote by $\cservs\subseteq \servs$ the set of certified services in the algorithm.
\end{definition}

\begin{proposition}
    \label{prop:OSD_EveryServiceCertifiedOnce}
    Every certified service $\serv \in \cservs$ is certified by exactly one other service.
\end{proposition}
\begin{proof}
    Suppose, for contradiction that $\serv$ is certified by two services, $\serv_1,\serv_2$, and assume WLOG that $\stime{\serv_1} < \stime{\serv_2}$.
    It must be that the triggering requests $\creq_{\serv_1}, \creq_{\serv_2}$ were witnesses for $\serv$ at $\stime{\serv_1},\stime{\serv_2}$, respectively.
    Thus, these requests were also of level $\level{\serv} + 1$, which implies $\level{\serv_1} = \level{\serv_2} =\level{\serv}+4$.

    We claim that after $\serv_1$, there remain no witness requests for $\serv$, in contradiction to $\creq_{\serv_2}$ being such a witness.
    Indeed, consider the state immediately before $\serv_1$: all witness requests for $\serv$ at that time have level $\level{\serv}+1$, and exist within the ball $\ball{a_\serv}{2^{\level{\serv}}}$ (as they were eligible for $\serv$).
    Both $\creq_{\serv_1},\creq_{\serv_2}$ are witness requests for $\serv$ at that time, and thus:
    Thus,
    \begin{align*}
        \dist{a_{\serv_1}}{\creq_{\serv_2}} &\le \dist{a_{\serv_1}}{\creq_{\serv_1}} + \dist{\creq_{\serv_1}}{a_{\serv}}+ \dist{a_{\serv}}{\creq_{\serv_2}} \\
        &\le  2^{\level{\serv_1} -3} + 2^{\level{\serv}} + 2^{\level{\serv}} \le 2^{\level{\serv_1}}.
    \end{align*}
    Therefore, $\creq_{\serv_2} \in \elig_{\serv_1}$; note that all requests in $\elig_{\serv_1}$ are either served in $\serv_1$ or become witnesses for $\serv_1$.
\end{proof}

\begin{proposition}
    \label{prop:OSD_ServiceCost}
    The cost of a service $\serv$ is at most $O(1) \cdot 2^{\level{\serv}}$.
\end{proposition}
\begin{proof}
    First, observe the cost of traversing the Steiner tree solution for the set of requests $\reqs_{\serv}$ as formed in the algorithm.
    Note that a possible Steiner tree solution for connecting $\reqs_{\serv} \cup \pc{a_{\serv}}$ would be to use the Steiner tree solution calculated in the penultimate iteration of the loop to connect all requests except $\req'$, where $\req'$ is the final request added to $\reqs_{\serv}$, then connect $\req'$ to $a_{\serv}$ directly.
    The cost of this solution is at most $4\cdot 2^{\level{\serv}}$ (from the condition of the loop) plus $\dist{a_{\serv}}{\req'}$ (which is at most $2^{\level{\serv}}$ since $\req' \in \elig_{\serv}$).
    Overall, the cost of this solution is at most $5\cdot 2^{\level{\serv}}$; thus, the cost of the Steiner tree chosen by the algorithm is at most $10\cdot 2^{\level{\serv}}$, as it uses a 2-approximation.
    The cost of traversing this tree from $\req$ is thus at most $20\cdot 2^{\level{\serv}}$.

    Second, the server possibly moves from its initial position $a_{\serv}$ to $\req$ (if $\serv$ is primary).
    The cost of this is at most $2^{\level{\serv}-1}$.
    Overall, the cost of a service $\serv$ is at most $O(1)\cdot 2^{\level{\serv}}$.
\end{proof}

\begin{lemma}
    \label{lem:OSD_AlgorithmBoundedByPrimaryAndCertified}
    $\alg \le O(1) \cdot \pr{\sum_{\serv \in \pservs} 2^{\level{\serv}} + \sum_{\serv \in \cservs} 2^{\level{\serv}}}$
\end{lemma}

\begin{proof}
    Using \cref{prop:OSD_ServiceCost}, we have that $\alg \le O(1) \cdot \sum_{\serv \in \servs} 2^{\level{\serv}} $.
    Consider any non-primary service $\serv \in \servs \backslash \pservs$.
    Since the service is non-primary, it certifies another service $\serv' \in \cservs$, such that $\level{\serv'} = \level{\serv} - 4$.
    Thus, we have that $2^{\level{\serv}} = 16\cdot 2^{\level{\serv'}}$; moreover, \cref{prop:OSD_EveryServiceCertifiedOnce} implies that services are only certified once, and thus $\sum_{\serv \in \servs\backslash \pservs} 2^{\level{\serv}} \le 16 \cdot \sum_{\serv' \in \cservs} 2^{\level{\serv'}}$.
    Overall, we have that
    \[
        \alg \le O(1) \cdot \sum_{\serv \in \servs} 2^{\level{\serv}} \le O(1) \cdot \pr{\sum_{\serv \in \pservs} 2^{\level{\serv}} + \sum_{\serv \in \cservs} 2^{\level{\serv}}}
    \]
\end{proof}

\subsubsection*{Charging Cylinders}

\Cref{lem:OSD_AlgorithmBoundedByPrimaryAndCertified} bounded the cost of the algorithm by the sum of two terms which correspond to primary and certified services.
It remains to charge those terms to the optimal solution.
To this end, we describe a method for charging costs to the optimal solution.

\textbf{Charging balls.}
Recall that $\ball{v}{r}$ denotes the ball of radius $r$ centered at some point $v \in \ms$.
Overloading notation, we use this terminology not only as a set of nodes, but also as a set of edges and ``parts'' of edges that exist within the ball; an informal visualization is given in \cref{subfig:OSD_BallIntersection}.
More formally, $\ball{v}{r}$ contains all edges where both endpoints are in $\ball{v}{r}$.
In addition, when an edge $e$ of weight $w_e$ has exactly one endpoint $u$ in $\ball{v}{r}$, the part of $e$ that belongs to $\ball{v}{r}$ is the segment of weight $r-\dist{v}{u}$ closest to $u$.
It is easy to see that this definition preserves desirable properties for edges in a ball; in particular, note that the edges and parts of edges in $\ball{v_1}{r_1}$ and in $\ball{v_2}{r_2}$ are disjoint if $\dist{v_1}{v_2} > r_1+r_2$.


\begin{figure}
    \centering
    \subfloat[][]{
        \includegraphics[width=0.45\columnwidth]{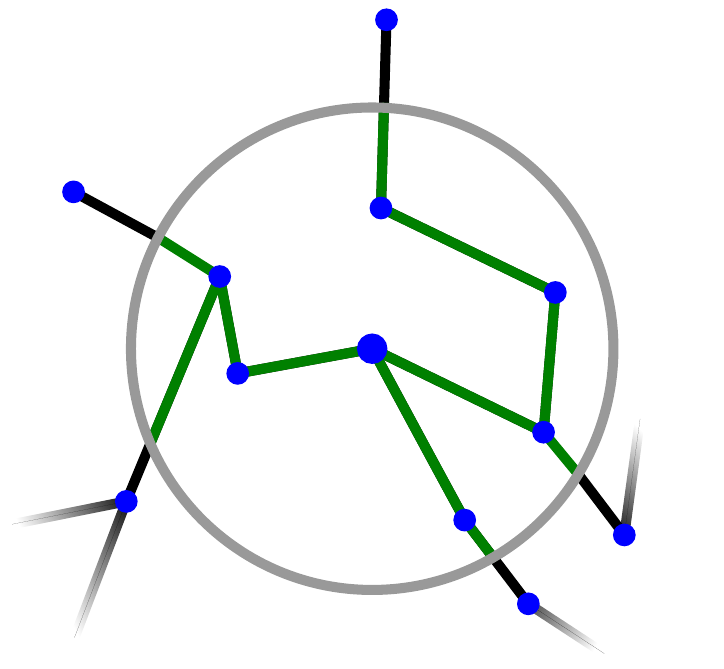}
        \label{subfig:OSD_BallIntersection}
    }
    \hfill
    \subfloat[][]{
        \includegraphics[width=0.45\columnwidth]{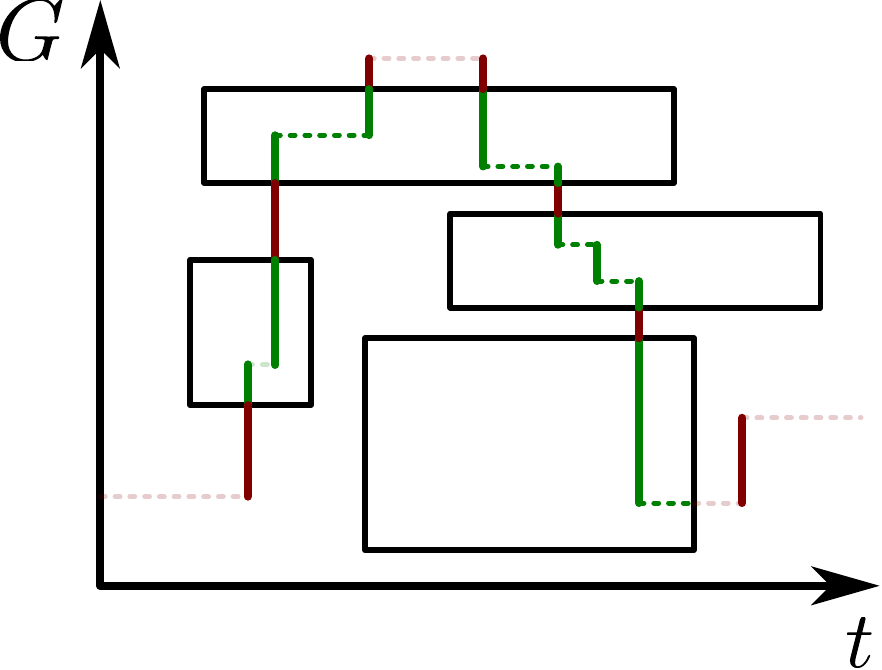}
        \label{subfig:OSD_CylinderIntersection}
    }
    \caption{Visualization of Intersections}
    \label{fig:OSD_Intersections}
\end{figure}

For the sake of charging costs, we create \emph{cylinders}.

\begin{definition}[cylinders]
    A cylinder is an ordered pair $(B,I)$ where $B$ is some shape in the metric space $\ms$, and $I$ is a time interval.
\end{definition}

As hinted by the word ``cylinder'', we later choose $B$ to be a ball (or a perforated ball, defined later).
The movement of the optimal solution inside the cylinder (i.e., during $I$ and inside $B$) is then charged to.

\begin{definition}[shape/cylinder charging]
    We use the following definitions to partition the costs of the optimal solution:
    \begin{enumerate}
        \item Given a subgraph $G' \subseteq \ms$ and a shape $B$, we denote by $\ccap{G'}{B}$ the total weight of edges (or parts of edges) in $G'$ that belong to $B$.
        \item Given a cylinder $\cyl=(B,I)$, define $\ccap{\opt}{\cyl}:=\ccap{G^*_I}{B}$, where $G^*_I$ is the subgraph of edges traversed by $\opt$ during $I$.
    \end{enumerate}
\end{definition}

For a set $\cyls$ of cylinders, define $\ccap{\opt}{\cyls} := \sum_{\cyl \in \cyls} \ccap{\opt}{\cyl}$ for ease of notation.
We now define disjointness for cylinders; a disjoint set of cylinders can charge to the optimal solution simultaneously, as they each charge to different movements of the server.

\begin{definition}[disjoint cylinders]
    A pair of cylinders $(B_1,I_1)$, $(B_2,I_2)$ are called disjoint if either $B_1$, $B_2$ are disjoint, or $I_1$, $I_2$ are disjoint.
\end{definition}

A set of cylinders whose metric shape is a ball can be seen in \cref{subfig:OSD_CylinderIntersection}.
Here, for the sake of visualization, we chose the metric space $\ms$ to be a line.
The cylinders thus appear as rectangles in the time-space plane.
The tour of the optimal server over time appears as a line, in which the dotted segments represent the passage of time and the solid segments show movement through space.
Only the length of solid segments inside a cylinder could be counted towards the intersection.
Thus, since the cylinders in the figure are disjoint, the total charged amount does not exceed the total moving cost of the optimal solution.
This is stated in \cref{obs:OSD_DisjointImpliesCharge}.

\begin{observation}
    \label{obs:OSD_DisjointImpliesCharge}
    Let $\cyls$ be a set of disjoint cylinders.
    Then
    \[
        \ccap{\opt}{\cyls} \le \opt
    \]
\end{observation}

With \cref{obs:OSD_DisjointImpliesCharge}, the way to use cylinders becomes clear: we want to construct a set of cylinders such that (a) their intersection with $\opt$ is large, and (b) they are disjoint (or can be partitioned into a few disjoint subsets).

\subsubsection*{Bounding Primary Services}

In this subsection, we focus on bounding the cost of primary services.
For every service $\serv$, define $a^*_{\serv}$ to be the final location of the optimum's server at $\stime{\serv}$.
%
Define $\pservsf \subseteq \pservs$ to be the primary services $\serv$ such that $\dist{a^*_{\serv}}{\creq_\serv} \ge 2^{\level{\serv} - 6}$.
\Cref{prop:OSD_BoundingPrimaryByFar} shows that to bound the cost of primary services, it is enough to bound the cost of $\pservsf$.

\begin{proposition}
    \label{prop:OSD_BoundingPrimaryByFar}
    $\sum_{\serv \in \pservs} 2^{\level{\serv}} \le O(1)\cdot \opt + O(1) \cdot \sum_{\serv \in \pservsf} 2^{\level{\serv}}$.
\end{proposition}

\begin{proof}
    Define the potential function $\phi(t) := 4 \dist{a(t)}{a^*(t)}$, where $a(t),a^*(t)$ are the server locations of the algorithm and the optimum at $t$, respectively.
    Note that the potential function equals $0$ at the beginning of the input, and can only take on positive values.
    For a service $\serv \in \pservs$, define $\Delta_{\serv}$ to be the increase in potential function by the service $\serv$ through the movement of the algorithm's server in $\serv$.
    Note that:
    \begin{enumerate}
        \item The only server movements in the algorithm are in primary services (other services return the server to its previous location).
        \item Increases in potential due to movements in $\opt$ sum to at most $4 \cdot \opt$.
    \end{enumerate}
    Therefore, we have the following:
    \begin{equation}
        \label{eq:OSD_PotentialArgument}
        \sum_{\serv \in \pservs} \dist{\creq_{\serv}}{a_{\serv}} \le 4 \opt + \sum_{\serv \in \pservs} \pr{\dist{\creq_{\serv}}{a_{\serv}} + \Delta_{\serv}}
    \end{equation}
    Consider a service $\serv \in \pservs\backslash\pservsf$, such that  $\dist{a^*(\stime{\serv})}{\creq_\serv} < 2^{\level{\serv} - 6}$.
    In addition, note that the fact that $\serv$ is primary implies that $\level{\serv} = \ceil{\log \dist{a_{\serv}}{\creq_{\serv}}} + 3$, and thus $\dist{a_{\serv}}{\creq_{\serv}} \ge 2^{\level{\serv} - 4}$; we have
    \begin{align*}
        \Delta_{\serv} &= 4\cdot(\dist{a^*_\serv}{\creq_\serv} - \dist{a^*_\serv}{a_{\serv}})\\
        &\le 4\cdot\pr{\dist{a^*_\serv}{\creq_\serv} -\dist{\creq_{\serv}}{a_{\serv}} + \dist{a^*_\serv}{\creq_{\serv}}} \\
        &\le 4 \cdot (-2^{\level{\serv} - 5})=-2^{\level{\serv}-3}
    \end{align*}
    Observing that $\dist{a_{\serv}}{\creq_{\serv}} \le 2^{\level{\req}-3}$, we have $\dist{a_{\serv}}{\creq_{\serv}} + \Delta_{\serv} \le 0$ for every $\serv \in \pservs\backslash\pservsf$.
    Moreover, note that for every $\serv \in \pservs$, it holds that $\Delta_{\serv} \le 4\dist{a_{\serv}}{\creq_{\serv}}$, and thus $\dist{a_{\serv}}{\creq_{\serv}} + \Delta_{\serv} \le 5\dist{a_{\serv}}{\creq_{\serv}}$.

    Combining all observations, we get
    \begin{align*}
        \sum_{\serv \in \pservs} 2^{\level{\serv}} &\le 16\sum_{\serv \in \pservs} \dist{\creq_{\serv}}{a_{\serv}} \\
        &\le 64\cdot \opt + 16 \sum_{\serv \in \pservs} \pr{\dist{\creq_{\serv}}{a_{\serv}} + \Delta_{\serv}} \\
        &\le 64\cdot \opt + 16 \sum_{\serv \in \pservsf} \pr{\dist{\creq_{\serv}}{a_{\serv}} + \Delta_{\serv}} \\
        &\le 64\cdot \opt + 80 \sum_{\serv \in \pservsf} \dist{\creq_{\serv}}{a_{\serv}} \\
        &\le 64\cdot \opt + 10 \sum_{\serv \in \pservsf} 2^{\level{\serv}} \\
    \end{align*}
    where the first inequality is from $\dist{\creq_{\serv}}{a_{\serv}} \ge {2^{\level{\serv}-4}}$, the second inequality is through \cref{eq:OSD_PotentialArgument}, the third inequality is from the fact that for every $\serv \in \pservs\setminus\pservsf$ we have $\dist{\creq_{\serv}}{a_{\serv}} + \Delta_{\serv} \le 0$, the fourth inequality is through $\dist{\creq_{\serv}}{a_{\serv}} + \Delta_{\serv} \le 5\dist{\creq_{\serv}}{a_{\serv}}$, and the final inequality is due to $\dist{\creq_{\serv}}{a_{\serv}} \le {2^{\level{\serv}-3}}$.
\end{proof}

To finish bounding the cost of primary services, it is thus enough to prove the following lemma.

\begin{lemma}
    \label{lem:OSD_BoundingFarPrimary}
    $\sum_{\serv \in \pservsf} 2^{\level{\serv}} \le O(\log \ar) \cdot \opt$.
\end{lemma}

\begin{definition}[primary interals and cylinders]
    For every service $\serv \in \pservsf$, we define:
    \begin{enumerate}
        \item The primary time interval $\pivl{\serv} := \I{\rlt{\creq_{\serv}}}{\dlt{\creq_{\serv}}}$.
        \item The primary cylinder $\pcyl{\serv} := \pr{\ball{a_{\serv}}{2^{\level{\serv}-2}}, \pivl{\serv}}$.
    \end{enumerate}
    We also define $\pcyls$ to be the set of all primary cylinders of services from $\pservsf$.
    In addition, for every $i$ we define $\pcyls_i$ to be the set of primary cylinders of level-$i$ services from $\pservsf$.
\end{definition}

\begin{proposition}
    \label{prop:OSD_PrimaryCylinderIntersection}
    For every $\serv \in \pservsf$, it holds that
    \[
        \ccap{\opt}{\pcyl{\serv}} \ge 2^{\level{\serv}-6}
    \]
\end{proposition}
\begin{proof}
    Since $\serv \in \pservsf$, we know that the server of the optimal solution was at $\creq_{\serv}$ somewhere during $\pivl{\serv}$, but was outside $\ball{\creq_{\serv}}{2^{\level{\serv}-6}}$ at $\stime{\serv}$; thus, the optimal solution incurred a cost of at least $2^{\level{\serv}-6}$ inside $\ball{\creq_{\serv}}{2^{\level{\serv}-6}}$ during $\pivl{\serv}$.
    Observing that $\ball{\creq_{\serv}}{2^{\level{\serv}-6}} \subseteq \ball{a_{\serv}}{2^{\level{\serv}-2}}$ implies that $\ccap{\opt}{\pcyl{\serv}} \ge 2^{\level{\serv}-6} $.
\end{proof}

\begin{proposition}
    \label{prop:OSD_PrimaryCylindersClassDisjoint}
    For every $i$, $\pcyls_i$ is a set of disjoint cylinders.
\end{proposition}
\begin{proof}
    Assuming otherwise, there exist $i$ and two level-$i$ services $\serv_1, \serv_2 \in \pservsf$ such that $\pcyl{\serv_1}, \pcyl{\serv_2}\in \pcyls_i$ are not disjoint.
    This implies that $\pivl{\serv_1}\cap \pivl{\serv_2} \neq \emptyset$; hence, WLOG, assume that $\stime{\serv_1} \in \pivl{\serv_2} = \I{\rlt{\creq_{\serv_2}}}{\dlt{\creq_{\serv_2}}}$.
    Thus, $\creq_{\serv_2}$ was pending during $\serv_1$, and moreover had level at most $ \level{\serv_2}$.
    But since $\pcyl{\serv_1}, \pcyl{\serv_2}$ are not disjoint, we have $\dist{a_{\serv_1}}{a_{\serv_2}} \le 2^{\level{\serv_1} -1}$, but this implies that
    \[
        \dist{a_{\serv_1}}{\creq_{\serv_2}} \le \dist{a_{\serv_1}}{a_{\serv_2}} + \dist{a_{\serv_2}}{\creq_{\serv_2}} \le 2^{\level{\serv_1} -1} + 2^{\level{\serv_1} -3} \le 2^{\level{\serv_1}}
    \]
    Thus, $\creq_{\serv_2} \in \elig_{\serv_1}$.
    But requests in $\elig_{\serv_1}$ are either served by $\serv_1$ or have their level increased to $\level{\serv_1}+1$, in contradiction to $\serv_2$ being primary.
\end{proof}

\begin{proof}
    [Proof of \cref{lem:OSD_BoundingFarPrimary}]
    The following holds:
    \begin{align*}
        \sum_{\serv \in \pservsf} 2^{\level{\serv}} &\le O(1)\cdot \sum_{\text{level }i} \ccap{\opt}{\pcyls_i} \\
        &\le O(1)\cdot \sum_{\text{level }i} \opt \\
        & O(\log \ar) \cdot \opt
    \end{align*}
    where the first inequality is due to \cref{prop:OSD_PrimaryCylinderIntersection}, the second inequality is due to \cref{prop:OSD_PrimaryCylindersClassDisjoint}, and the third inequality is due to the fact that there are only $O(\log \ar)$ possible classes for primary services.
\end{proof}

\subsubsection*{Bounding Certified Services}

In this subsection, we focus on bounding the cost of certified services; specifically, we prove the following lemma.

\begin{lemma}
    \label{lem:OSD_BoundingCertified}
    $\sum_{\serv \in \cservs} 2^{\level{\serv}} \le O(\log (\ar\npoint)) \cdot \opt$.
\end{lemma}

\begin{definition}[$\ptime{\serv}$ and $\civl{\serv}$]
    Let $\serv \in \cservs$ be a certified service.
    Let $\serv' \in\cservs$ be the certified service with maximum $\ftime{\serv'}$ subject to $\level{\serv'} = \level{\serv}$, $\ftime{\serv'} \le \stime{\serv}$ and $\dist{a_{\serv'}}{a_{\serv}} \le 6\cdot 2^{\level{\serv}}$.
    We define:
    \begin{enumerate}
        \item The time $\ptime{\serv} := \ftime{\serv'}$ if $\serv'$ exists (otherwise, define $\ptime{\serv} = -\infty$).
        \item The time interval $\civl{\serv} := \I{\ptime{\serv}}{\ftime{\serv}}$; note that $\stime{\serv}\in\civl{\serv}$.
    \end{enumerate}
\end{definition}

\begin{proposition}
    \label{prop:OSD_CertifierBetweenCertified}
    Let $\serv_1, \serv_2 \in \cservs$ be such that $\level{\serv_1} = \level{\serv_2} = \ell$ and $\dist{a_{\serv_1}}{a_{\serv_2}} \le 6\cdot 2^{\ell}$.
    Assuming WLOG that $\stime{\serv_1} < \stime{\serv_2}$, and letting $\serv$ be the level-$(\level{\serv_1}+4)$ service that made $\serv_1$ certified, it holds that $\stime{\serv}\in\I{\ftime{\serv_1}}{\stime{\serv_2}}$.
    (In particular, $\ftime{\serv_1} < \stime{\serv_2}$.)
\end{proposition}
\begin{proof}
    The two possible cases which contradict our proposition are that $\stime{\serv} \le \ftime{\serv_1}$ or that $\stime{\serv} >\stime{\serv_2}$.
    If $\stime{\serv} \le \ftime{\serv_1}$, consider the triggering request $\creq_{\serv}$: this request was a witness for $\serv_1$, and thus in $\elig_{\serv_1}$; however, $\serv_1$ chose which requests from $\elig_{\serv_1}$ to serve according to earliest deadline, and managed to serve all requests in $\elig_{\serv_1}$ of deadline $\le \ftime{\serv_1}$ (by definition).
    The existence of $\creq_{\serv_1}$ as a pending request at $\stime{\serv}$ is therefore a contradiction.

    Otherwise, if $\stime{\serv} >\stime{\serv_2}$, consider the service $\serv'$ which certified $\serv_2$; it must also be the case that $\stime{\serv'} > \stime{\serv_2}$.
    Suppose that $\stime{\serv} < \stime{\serv'}$.
    In this case, observe that all witnesses for $\serv_2$ at $\stime{\serv}$ are in $\elig_{\serv}$; therefore, these witnesses would no longer be witnesses for $\serv_2$ after $\serv$, in contradiction to one of them triggering $\serv'$ and certifying $\serv_2$.
    Similarly, if $\stime{\serv} >\stime{\serv'}$, the service $\serv'$ would leave no witnesses for $\serv_1$ to trigger $\serv$.
    We thus again reached a contradiction.
\end{proof}

\begin{definition}[certified cylinders]
    For a certified service $\serv \in \cservs$, define the certified cylinder $\ccyl{\serv} := (\ball{\creq_{\serv}}{3\cdot 2^{\level{\req}}}, \civl{\serv})$.
    Define $\ccyls$ to be the set of all certified cylinders; in addition, for every $i$ define $\ccyls_i$ to be the set of certified cylinders formed from level-$i$ services.
\end{definition}

\begin{proposition}
    \label{prop:OSD_CertifiedCylindersDisjoint}
    For every $i$, the set $\ccyls_i$ is a set of disjoint cylinders.
\end{proposition}
\begin{proof}
    Consider any two cylinders $\ccyl{\serv_1}, \ccyl{\serv_2} \in \ccyls_i$.
    If it holds that $\dist{a_{\serv_1}}{a_{\serv_2}} > 6\cdot 2^{i}$, then the cylinders are spatially disjoint and we are done.
    Thus, assume that $\dist{a_{\serv_1}}{a_{\serv_2}} \le 6\cdot 2^{i}$, and WLOG assume that $\stime{\serv_1} < \stime{\serv_2}$.
    \Cref{prop:OSD_CertifierBetweenCertified} implies that $\ftime{\serv_1} < \stime{\serv_2}$; from the definition of $\ptime{\serv_2}$, it is thus also the case that $\ptime{\serv_2} \ge \ftime{\serv_1}$.
    Thus, the intervals $\civl{\serv_1},\civl{\serv_2}$ are disjoint, and thus $\ccyl{\serv_1}, \ccyl{\serv_2}$ are disjoint.
\end{proof}

Having defined the certified cylinders and shown a disjointness property in \cref{prop:OSD_CertifiedCylindersDisjoint}, we want to show that the optimal solution has a large intersection with these cylinders.
We show this by claiming that the release-to-deadline intervals for requests in $\elig_{\serv}$ are contained in $\civl{\serv}$, and thus must be served by the optimal solution in this interval.
Therefore, a Steiner tree spanning $\elig_{\reqs}$ can be charged to the optimal solution in this time interval.
However, one must still claim that enough of this cost takes place within the ball defining the cylinder.
This is possible as the requests of $\elig_{\serv}$ are in $\ball{a_{\serv}}{2^{\level{\serv}}}$, while the radius of $\ccyl{\serv}$ is $3\cdot 2^{\level{\serv}}$.

\begin{figure}
    \centering
    \subfloat[][]{
        \includegraphics[width=0.45\columnwidth]{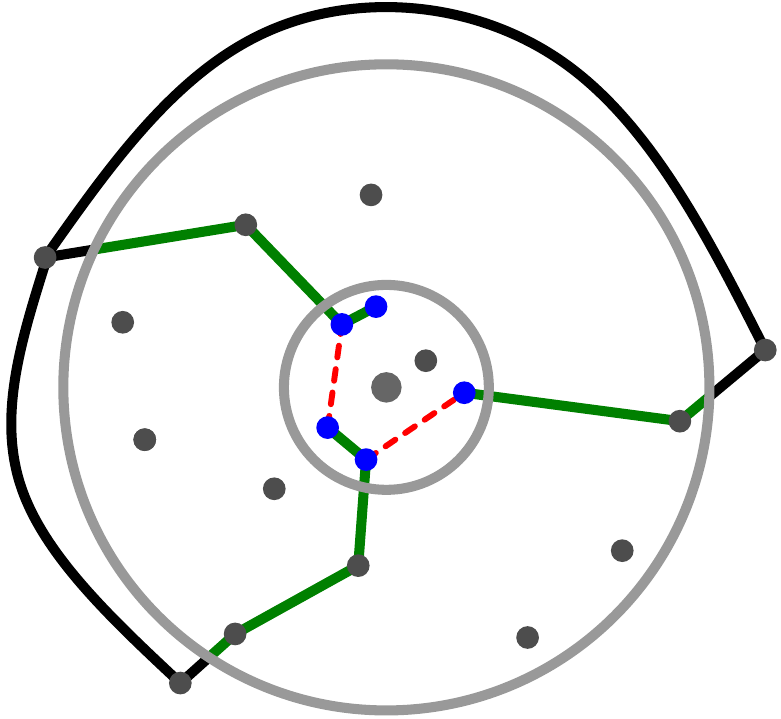}
        \label{subfig:OSD_SteinerBall}
    }
    \hfill
    \subfloat[][]{
        \includegraphics[width=0.45\columnwidth]{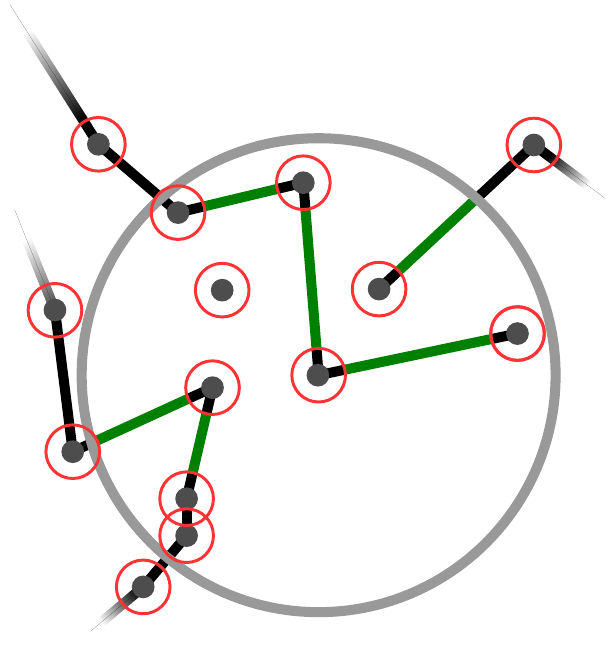}
        \label{subfig:OSD_Perforation}
    }
    \caption{Visualizations from the analysis of \cref{alg:OSD}.}
    \label{fig:OSD_SteinerBall}
\end{figure}

\begin{proposition}
    \label{prop:OSD_BallIntersection}
    Consider a set of points $V \subseteq \ball{\rho}{r}$, and let $G'$ be a subgraph of $\ms$ that connects $V$.
    Then it holds that
    \[
        \stree^*(V) \le 2 \cdot \ccap{G'} {\ball{\rho}{3r}}
    \]
\end{proposition}

\begin{proof}
    Consider the edges in $G'$ contained in $\ball{\rho}{3r}$.
    If those edges connect all points in $V$, then they form a valid solution for Steiner tree on $V$, and thus we are done.
    Otherwise, these edges partition $V$ into connected components $V_1, \cdots, V_k$.
    Since $V$ are connected in $G'$, these connected components must be connected somewhere outside $\ball{\rho}{r}$; in particular, connected component contains a path which exits $\ball{\rho}{3r}$; since the connected component contains a point in $V \subseteq \ball{\rho}{r}$, the total weight of this component is at least $2r$.

    We can convert $G\cap \ball{\rho}{3r}$ into a Steiner tree solution for $V$ by connecting at most $k-1$ pairs of points from $V$ directly, such that the points of each pair belong to different components $V_i, V_j$.
    Since the diameter of $V$ is at most $2r$, and the cost of each connected component is at least $2r$, this modification at most doubles the cost of $G'\cap \ball{\rho}{3r}$.
    This completes the proof.
\end{proof}

A visual description of the proof of \cref{prop:OSD_BallIntersection} is given in \cref{subfig:OSD_SteinerBall}.
Here, the terminals $V$ (in blue) are connected by the subgraph $G'$ such that the restriction to $\ball{\rho}{3r}$ creates three connected components.
As the terminals are all in $\ball{\rho}{r}$, adding the two red, dashed edges would augment this restriction to a Steiner tree solution for $V$, where each edge costs at most the diameter of the smaller ball (i.e., $2r$).

\begin{proposition}
    \label{prop:OSD_RequestInCertifiedCylinder}
    For every certified service $\serv \in \cservs$ and request $\req \in \reqs_{\serv}$, it holds that $\I{\rlt{\req}}{\dlt{\req}} \subseteq \I{\ptime{\serv}}{\ftime{\serv}}$.
\end{proposition}

\begin{proof}
    Define $\ell = \level{\req}$.
    By definition, it holds that $\dlt{\req}\le \ftime{\serv}$.
    It remains to show that $\rlt{\req} \ge \ptime{\serv}$.
    If $\ptime{\serv} = -\infty$, we are done; otherwise, by the definition of $\ptime{\serv}$, there exists a certified service $\serv'$ such that $\level{\serv'} = \ell$, $\ptime{\serv} = \ftime{\serv'}$ and $\dist{a_{\serv'}}{a_{\serv}} \le 6\cdot 2^{\level{\serv}}$.
    Applying \cref{prop:OSD_CertifierBetweenCertified}, there exists a level-$(\ell+4)$ service $\serv''$ in $\I{\ftime{\serv'}}{\stime{\serv}}$, such that $\dist{\creq_{\serv''}}{a_{\serv'}} \le 2^{\ell}$.

    Now, observe that
    \begin{align}
        \dist{\req}{a_{\serv''}} &\le \dist{\req}{a_{\serv}} + \dist{a_{\serv}}{a_{\serv'}}+\dist{a_{\serv'}}{\creq_{\serv''}} + \dist{\creq_{\serv''}}{a_{\serv''}}\nonumber \\
        \label{eq:OSD_RequestCloseToCertifier}&\le  2^{\level{\serv}} + 6\cdot 2^{\level{\serv}} + 2^{\level{\serv}} + 2^{\level{\serv}+1}  = 10\cdot 2^{\level{\serv}} \le 2^{\level{\serv''}}
    \end{align}

    Assume for contradiction that $\rlt{\req} < \ptime{\serv}$, which implies that $\rlt{\req} < \ftime{\serv'}$.
    Since $\req \in \reqs_{\serv}$, we have that $\req$ was pending during $\I{\ftime{\serv'}}{\stime{\serv}}$.
    Thus, $\req$ was pending during $\serv''$; combining with \cref{eq:OSD_RequestCloseToCertifier}, it must be that $\req \in \elig_{\serv''}$.
    But then it must be that $\serv''$ raised the level of $\req$ to $\level{\serv''} + 1 = \level{\serv} + 5$, in contradiction to $\req \in \reqs_{\serv} \subseteq \elig_{\serv}$.
\end{proof}

\begin{proposition}
    \label{prop:OSD_CertifiedCylinderIntersection}
    For every certified cylinder $\ccyl{\serv}$, it holds that $ 2^{\level{\serv}-1} \le \ccap{\opt}{\ccyl{\serv}} $.
\end{proposition}
\begin{proof}
    Consider the union of edges traversed by the optimum during $\I{\ptime{\serv}}{\ftime{\serv}}$, and denote it by $G^*$;
    note that $\ccap{\opt}{\ccyl{\serv}} = \ccap{G^*}{\ball{a_\serv}{3\cdot 2^{\level{\serv}}}}$.
    From \cref{prop:OSD_RequestInCertifiedCylinder}, $G^*$ must connect $\reqs_{\serv}$; Noting that $\reqs_\serv \subseteq \ball{a_\serv}{2^{\level{\serv}}}$, \cref{prop:OSD_BallIntersection} implies
    $ \stree^*(\reqs_{\serv}) \le 2\cdot  \ccap{G^*}{\ball{\creq_\serv}{3\cdot 2^{\level{\serv}}}}$.
    Now, note that since $\serv$ is a certified service, \cref{line:OSD_TreeReachedBudget} was reached, and thus the approximate solution of the algorithm for $\stree(\reqs_{\serv} \cup \pc{a_{\serv}})$ had cost at least $4\cdot 2^{\level{\serv}}$.
    Since the Steiner-tree algorithm is a 2-approximation, we have $\stree^*(\reqs_{\serv} \cup \pc{a_{\serv}}) \ge 2\cdot 2^{\level{\serv}}$.
    Finally, since $a_{\serv}$ can be connected directly to any request in $\reqs_{\serv}$ at cost at most $2^{\level{\serv}}$, we have $\stree^*(\reqs_{\serv} \cup \pc{a_{\serv}}) \le \stree^*(\reqs_{\serv}) + 2^{\level{\serv}}$, which therefore implies $\stree^*(\reqs_{\serv}) \ge 2^{\level{\serv}}$.
    Combining, we have that $2^{\level{\serv}-1} \le \ccap{\opt}{\ccyl{\serv}}$.
\end{proof}

\begin{proof}
    [Proof of \cref{lem:OSD_BoundingCertified}]
    The following holds:
    \begin{align*}
        \sum_{\serv \in \cservs} 2^{\level{\serv}} &\le O(1) \cdot \sum_{\text{level }i} \ccap{\ccyls_i}{\opt} \\
        &\le O(1)\cdot \sum_{\text{level }i} \opt \\
        &\le O(\log (\ar\npoint))\cdot \opt
    \end{align*}
    where the first inequality is due to \cref{prop:OSD_CertifiedCylinderIntersection}, the second inequality is due to \cref{prop:OSD_CertifiedCylindersDisjoint}, and the final inequality is due to the fact that the number of possible classes for services is $O(\log (\ar\npoint))$.
\end{proof}

\begin{proof}
    [Proof of \cref{thm:OSD_WeakCompetitiveness}]
    We have
    \begin{align*}
        \alg &\le O(1) \cdot \pr{\sum_{\serv \in \cservs} 2^{\level{\serv}} +\sum_{\serv \in \pservsf} 2^{\level{\serv}} + \opt  } \\
        &\le O(\log (\ar\npoint)) \cdot \opt
    \end{align*}
    where the first inequality is due to \cref{lem:OSD_AlgorithmBoundedByPrimaryAndCertified,prop:OSD_BoundingPrimaryByFar}, and the second inequality is due to  \cref{lem:OSD_BoundingFarPrimary,lem:OSD_BoundingCertified}.
\end{proof}

\subsubsection*{Improved Analysis through Perforated Cylinders}

We now alter our cylinder construction to obtain \cref{thm:OSD_Competitiveness}.
We do so by changing the shape of cylinders in the metric space from a ball to a \emph{perforated ball}.
A perforated ball used in our proofs is formed from a ball of radius $r$ by removing balls of radius $\frac{r}{c\cdot \npoint^2}$ around every point in the metric space, for some constant $c > 1$.
Since the radii of removed balls are small, we claim that the intersection of cylinders using these new, perforated balls with $\opt$ is only smaller by a constant factor from the original intersection.
However, the perforation ensures that cylinders with an $\Omega(\log \npoint)$ gap do not intersect, yielding increased disjointness and better competitiveness.
An informal visualization of perforated balls appears in \cref{subfig:OSD_Perforation}, which shows the intersection of a subgraph with such a ball.
A formal description, as well as the proofs of \cref{lem:OSD_BoundingFarPrimaryImproved,lem:OSD_BoundingCertifiedImproved} appear in \cref{sec:Perf}.

\begin{lemma}
    [improved \cref{lem:OSD_BoundingFarPrimary}]
    \label{lem:OSD_BoundingFarPrimaryImproved}
    \[
        \sum_{\serv \in \pservsf} 2^{\level{\serv}} \le O(\log \npoint) \cdot \opt.
    \]
\end{lemma}

\begin{lemma}
    [improved \cref{lem:OSD_BoundingCertified}]
    \label{lem:OSD_BoundingCertifiedImproved}
    \[
        \sum_{\serv \in \cservs} 2^{\level{\serv}} \le O(\log \npoint) \cdot \opt.
    \]
\end{lemma}

\begin{proof}
    [Proof of \cref{thm:OSD_Competitiveness}]
    The proof of the theorem is identical to that of \cref{thm:OSD_WeakCompetitiveness}, except that \cref{lem:OSD_BoundingFarPrimary,lem:OSD_BoundingCertified} are replaced with \cref{lem:OSD_BoundingFarPrimaryImproved,lem:OSD_BoundingCertifiedImproved}.
\end{proof}

    \section{Conclusions and Future Directions}
    \label{sec:Conc}
    In this paper, we introduced the first deterministic algorithms for online service with deadlines/delay for general metric spaces (of subpolynomial competitiveness).
Our algorithms also improve upon the best known randomized algorithms for the problem.
While superconstant lower bounds for this problem are not yet known, there is evidence that the $O(\log \npoint)$ competitiveness shown in this paper is tight.
This is suggested by the fact that previous to our work, $O(\log \npoint)$-competitiveness was the best known bound for both multilevel aggregation (see~\cite{DBLP:conf/focs/AzarT20}) and online service with delay on an equidistant line (see~\cite{DBLP:conf/sirocco/BienkowskiKS18}), which are both special cases of the problem considered in this paper.
Introducing a superconstant lower bound for this problem would thus be a great future direction.
Also note that currently, no separation is known between deterministic and randomized algorithms: for service with deadlines/delay, as well as for the mentioned special cases, the best known algorithms are deterministic, and there exists no known superconstant lower bound even for deterministic algorithms.

In addition, we believe that our techniques could be extended to multiple servers, yielding a deterministic algorithm for $k$-server with delay.
This would require some combination of our techniques with existing deterministic techniques for $k$-server on general metric spaces, e.g., the work function algorithm~\cite{DBLP:journals/jacm/KoutsoupiasP95}.

    \appendix

    \section{Online Service with Delay}
    \label{sec:OSY}
    This section considers the online service with delay problem.
For this problem, we prove the following theorem.

\begin{theorem}
    \label{thm:OSY_Competitiveness}
    There exists a $O(\log \npoint)$-competitive deterministic algorithm for online service with delay.
\end{theorem}

\subsection{The Algorithm}

\textbf{Prize-Collecting Steiner tree.}
Similar to the algorithm for deadlines, the algorithm for delay also requires an approximation algorithm for Steiner tree.
However, we now consider the prize-collecting variant of the Steiner tree problem.
In this variant, we are given a set of terminals and a root node; in addition, each terminal has an associated \emph{penalty}.
A solution is a subgraph that connects some subset of the terminals to the root node.
The cost of that solution is the total weight of the subgraph, plus the penalties for terminals that were not connected to the root node.

There exists a 3-approximation for prize-collecting Steiner tree, due to Hajiaghayi et al.~\cite{Hajiaghayi2006}, which we denote $\pcstree$.
We use $\pcstree(U, \pi; r)$ to denote running the approximation over the input graph $\ms$, with terminal set $U$, penalty function $\pi$ (that maps from terminal to its penalty), and root node $r$.
As before, we use $\pcstree(U, \pi; r)$ to refer to the cost of the approximate solution (edge and penalty cost).
We also use $\pcstree^*(U,\pi;r)$ to refer to the optimal solution for the same input (and its cost).

\textbf{Algorithm's description.}
As in the deadline algorithm, every pending request $\req$ has a level $\level{\req}$ (which is initially $-\infty$), and we define the adjusted level of $\req$ to be $\alevel{\req} := \max\pc{\level{\req}, \ceil{\log \dist{\req}{a}}}$, where $a$ is the current location of the server.
In addition, for every request $\req$ the algorithm maintains an investment counter $\ctr_{\req}$ (which is initially $0$).
This counter is raised by a service to pay for (past or future) delay of a request.
We denote by $\level{\req}(t),\alevel{\req}(t), \ctr_{\req}(t)$ the values of $\level{\req},\alevel{\req},\ctr_{\req}$ at time $t$ (if a service takes place at $t$, this refers to the values immediately before the service).
We also define the \emph{residual delay} of $\req$ at $t$ to be $\rylt{\req}{t} := \prp{\ylt{\req}{t} - \ctr_{\req}(t)}$; intuitively, this is the amount of current delay which no service has paid for.

The following definition of a critical level is used to trigger services in the algorithm.

\begin{definition}[critical level]
    Fix any time $t$, and a level $\ell$.

    \begin{enumerate}
        \item Define $\trd_\ell(t)$ to be the total residual delay of requests $\req$ s.t. $\alevel{\req} \le \ell$.
        \item We say that $\ell$ is \emph{critical} if $\trd_\ell(t) \ge 2^{\ell}$.
    \end{enumerate}
\end{definition}

The pseudocode for the algorithm is given in \cref{alg:OSY}.
Whenever a level $\ell$ becomes critical, the function $\UponCritical$ is called, which initiates a service $\serv$;
the level of this service $\level{\serv}$ is set to be $\ell+3$.
The service identifies the triggering set of requests $\creqs_{\serv}$, which are the requests whose total residual delay became critical (that is, requests of adjusted level at most $\level{\serv}-3 = \ell$ with positive residual delay).
If there exists no triggering request of level at least $\level{\serv}-4$ (i.e., at least $\ell - 1$), the service is called primary.

For a primary service $\serv$, the algorithm attempts to identify a new location for placing the server after the service concludes.
If a constant fraction of the residual delay of $\creqs_{\serv}$ exists inside a small radius ball, the center of that ball will be the new location of the server.
Otherwise, the final location of the server will be its starting location.

A service $\serv$ identifies all requests eligible to the service, which are all requests with adjusted level at most $\level{\serv}$; these are denoted $\elig_{\serv}$.
The service then resets the residual delay of all requests in $\elig_{\serv}$.
(Note, in particular, that this resets the residual delay of all triggering requests in $\creqs_{\serv}$.)
Then, the service performs time forwarding: it observes the future delay of requests in $\elig_{\serv}$, and attempts to ``pay'' for those requests until a point in time furthest in the future.
``Paying'' for a request $\req$ until time $\tau$ can be done either through serving that request, or by increasing its investment counter to at least $\ylt{\req}{\tau}$.
Choosing between these options is done through calling $\pcstree$ on the eligible requests with a penalty function which represents the required increase to investment counters.
Specifically, the algorithm finds the first time $\tau$ in which the cost of $\pcstree$ exceeds $c\cdot 2^{\level{\serv}}$, for some constant $c$; the algorithm will serve requests and increase counters according to the $\pcstree$ solution for time $\tau$.

Finally, at the end of the service, eligible requests that are still pending are upgraded to a level higher than that of the service, and the server moves to its new final location (if applicable).

\begin{algorithm}
    \caption{Online Service with Delay}
    \label{alg:OSY}
    \EFn{\UponRequest{$\req$}}{
        set $\level{\req} \gets -\infty$, $\ctr_\req \gets 0$.
    }
    \EFn{\UponCritical{$\ell$}}{
        start a new service, denoted by $\serv$, and set $\level{\serv} \gets \ell +3$.

        let $t$ be the time, $a$ be the server's location, and $\reqs'$ be the set of pending requests.

        let $\creqs_{\serv} \gets \pc{\req' \in \reqs' \middle| \alevel{\req} \le \level{\serv} - 3 \wedge \rylt{\req}{t} > 0}$.
        \label{line:OSY_DefineTriggeringRequests}

        \lIfElse{for every $\req' \in \creqs_{\serv}$ we have $\level{\req'} < \level{\serv} - 4$}{say $\serv$ is primary}{$\serv$ is not primary.}
        \label{line:OSY_SetPrimary}

        \eIf{$\serv$ is primary \And there exists $a'\in\ms$ s.t. $\rylt{R}{t} > 2^{\level{\serv}-4}$ where $R := \creqs_{\serv} \cap \ball{a'}{2^{\level{\serv}-8}}$}{define $a'$ as mentioned.}{set $a' \gets \textsf{Null}$.}
        \label{line:OSY_FindNewLoc}

        set $\elig_{\serv} \gets \pc{\req' \in \reqs' \middle| \alevel{\req} \le \level{\serv}}$.
        \label{line:OSY_DefineEligible}

        \lForEach{$\req' \in \elig_{\serv}$}{set $\ctr_{\req'} \gets \max\pc{\ctr_{\req'}, \ylt{\req'}{t}}$.\label{line:OSY_ZeroResidualDelay}\quad \tcp*[h]{Zero residual delay.}}

        set $\reqs_{\serv} \gets \pc{\req}, S\gets \emptyset$.

        for every time $t' \ge t$, define the penalty function $\pen_{t'}:\elig_{\serv} \to \mathbb{R}^+\cup\pc{0}$ such that $\pen_{t'}(\req') = \prp{\ylt{\req'}{t'} - \ctr_{\req'}}$ for every $\req' \in \elig_{\serv}$.

        let $\tau \ge t$ be the first time in which $\pcstree(\elig_{\serv}, \pen_{\tau}; a) \ge 6\cdot 2^{\level{\serv}}$.
        \label{line:OSY_DefineForwardingTime}

        let $S$ to be the solution $\pcstree(\elig_{\serv}, \pen_{\tau}; a)$, and let $\reqs_{\serv}\subseteq \elig_{\serv}$ be the set of requests served by $S$.

        perform DFS tour of $S$, serving $\reqs_{\serv}$ and finishing at $a$.\label{line:OSY_TraverseTree}

        \ForEach{$\req' \in \elig_{\serv}\backslash \reqs_{\serv}$}{
            set $\ctr_{\req'} \gets \max\pc{\ctr_{\req'}, \ylt{\req'}{\tau}}$.\label{line:OSY_Invest}

            set $\level{\req'} \gets \level{\serv} + 1$.\label{line:OSY_LevelUpgrade}
        }

        \lIf{$a' \neq \textsf{Null}$}{move the server from $a$ to $a'$.}\label{line:OSY_MoveServer}
    }
\end{algorithm}

A level-$\ell$ service is triggered when level $\ell-3$ becomes critical, i.e., when requests of adjusted level at most $\ell-3$ gather large residual delay.
\Cref{fig:OSY_TriggeringResidualDelay} gives some intuition about how the residual delay of those triggering requests are distributed.
In \cref{fig:OSY_TriggeringResidualDelay}, the delay of the triggering requests of a service $\serv$ is shown as a heatmap inside the ball $\ball{a_{\serv}}{2^{\level{\serv}-3}}$ (more delay is a deeper shade of red).
\Cref{subfig:OSY_TriggeringResidualDelaySecondary} shows a pattern that can only belong to a nonprimary service: in primary services, the outer ring $\ball{a_{\serv}}{2^{\level{\serv} -3}} \setminus \ball{a_{\serv}}{2^{\level{\serv}-5}}$ must contain a constant fraction of the residual delay (as shown in the proof of \cref{prop:OSY_MovingDistanceBounds}).
In secondary services, the server's final location is the same as its initial location.
\Cref{subfig:OSY_TriggeringResidualDelayPrimaryStationary} shows a delay pattern which might belong to a primary service.
If this is the case, the server will again finish at its initial location, as delay is not concentrated in a low-radius ball ($a'$ is not defined in \cref{line:OSY_FindNewLoc}).
\Cref{subfig:OSY_TriggeringResidualDelayPrimaryMove} shows a delay pattern which is highly concentrated in a low-radius ball.
Thus, if this pattern belongs to a primary service, the server would move to the center of the low-radius ball at the end of the service.

\begin{figure}

    \centering
    \subfloat[][]{
        \label{subfig:OSY_TriggeringResidualDelaySecondary}
        \includegraphics[width=0.45\columnwidth]{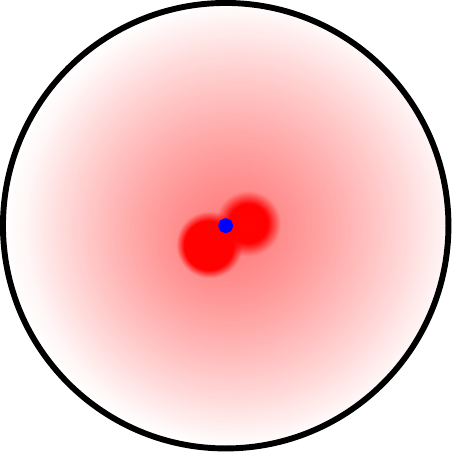}
    }
    \hfill
    \subfloat[][]{
        \label{subfig:OSY_TriggeringResidualDelayPrimaryStationary}
        \includegraphics[width=0.45\columnwidth]{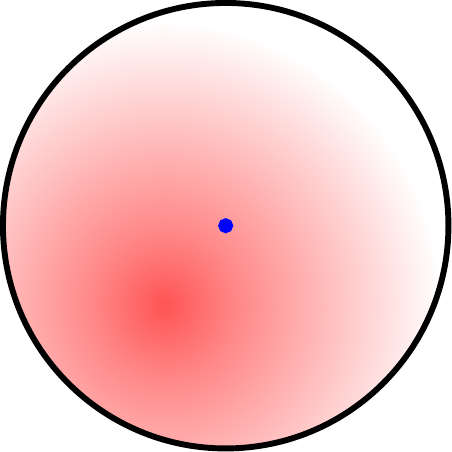}
    }
    \hfill
    \subfloat[][]{
        \label{subfig:OSY_TriggeringResidualDelayPrimaryMove}
        \includegraphics[width=0.45\columnwidth]{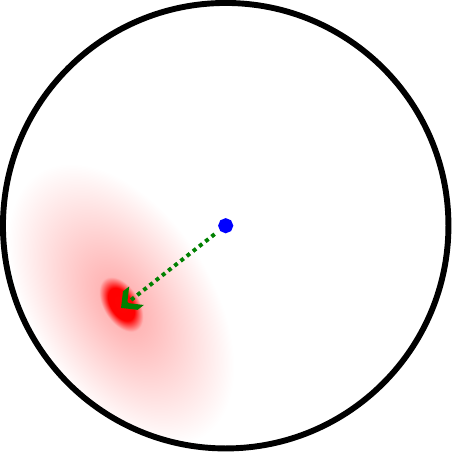}
    }
    \caption{Possible distributions for residual delay of triggering requests.}
    \label{fig:OSY_TriggeringResidualDelay}
\end{figure}

\subsection{Analysis}

We now focus on proving \cref{thm:OSY_Competitiveness}.

\begin{definition}[basic service definitions]
    Let $\serv$ be a service.
    We define:
    \begin{itemize}
        \item The \emph{triggering requests} $\creqs_{\serv}$ to be as defined in \cref{line:OSY_DefineTriggeringRequests}.
        \item The location $a_{\serv}$ to be the initial location of the server when $\serv$ is triggered.
        \item The term $\stime{\serv}$ to be the time of service $\serv$.
        \item The request set $\reqs_{\serv}$ to be the requests served by $\serv$ (i.e., the final value of the variable of that name in $\UponDeadline$).
        \item The request set $\elig_{\serv}$ as defined in \cref{line:OSY_DefineEligible}; these requests are called \emph{eligible for $\serv$}.
        \item The \emph{forwarding time} of $\serv$, denoted $\ftime{\serv}$, to be $\tau$ as defined in \cref{line:OSY_DefineForwardingTime}.
    \end{itemize}
\end{definition}

\begin{definition}[witness requests and certified services]
    \label{defn:OSY_WitnessAndCertified}
    At a certain point in time, a request $\req$ is called a \emph{witness} for a service $\serv$ if its level $\level{\req}$ was last assigned to by $\serv$ (at \cref{line:OSD_LevelUpgrade}).

    For every non-primary service $\serv'$, let $\req \in \creqs_{\serv'}$ be an arbitrary request such that $\level{\req}(\stime{\serv}) \ge \level{\serv'} -4$, and let $\serv$ be the service for which $\req$ is a witness.
    We say that $\serv'$ \emph{certifies} $\serv$, and call $\serv$ a \emph{certified} service.
    (Note that the chosen request $\req$ is unique, such that every service certifies at most one other service.)
\end{definition}

\begin{proposition}
    \label{prop:OSY_CertificationLevelDifference}
    If $\serv'$ certifies $\serv$, then $\level{\serv} + 4 \le \level{\serv'} \le \level{\serv} + 5$.
\end{proposition}
\begin{proof}
    Consider the request $\req \in \creqs_{\serv'}$ which was a witness for $\serv$. It holds that $\level{\req}(\stime{\serv'}) = \level{\serv} + 1 $.
    In addition, it holds that $\level{\req}(\stime{\serv'}) \le \level{\serv'} - 3$ (as $\req \in \creqs_{\serv'}$) and that $\level{\req}(\stime{\serv'}) \ge \level{\serv'} - 4$ (as $\req$ made $\serv'$ certify $\serv$).
\end{proof}

\begin{proposition}
    \label{prop:OSY_NoSuperCritical}
    Throughout the algorithm, for every level $i$ it holds that $\trd_{i} \le 2^{i+1}$ at every point in time.
\end{proposition}
\begin{proof}
    Continuous delay growth cannot break this proposition, as the algorithm triggers a service whenever a level $i$ becomes critical (i.e., when $\trd_{i}$ reaches $2^{i}$), and this service then zeroes $\trd_{i}$.
    The only conceivable way for this to happen is upon moving the server from $a$ to $a'$ in \cref{line:OSY_MoveServer}; indeed, this changes the adjusted levels of requests, possibly increasing $\trd_{i}(t)$ past $2^{i}$.
    However, we claim that this cannot increase $\trd_{\ell}$ too much.

    The server moved from $a$ to $a'$ during some service $\serv$, triggered by some level $\ell$ becoming critical; it holds that $\level{\serv} = \ell + 3$.
    Let $t^{-}, t^{+}$ be the times immediately before and after \cref{line:OSY_MoveServer} in $\serv$.
    At $t^{-}$, there are no requests with adjusted level less than $\ell+4$, as ensured by $\cref{line:OSY_LevelUpgrade}$; in particular, $\trd_{i} = 0 $ for every $i \le \ell + 3$.
    Moreover, it holds that $\trd_{i} \le 2^i$ for every $i> \ell+3$, as only the maximal critical level triggers a service.
    Thus, $\trd_i(t^-) \le 2^{i}$ for all $i$.

    Consider any pending request $\req$ at $t^-$; as mentioned, $\alevel{\req}(t^-) \ge \ell + 4$.
    We claim that $\alevel{\req}(t^+) \ge  \alevel{\req}(t^-) -1 $; if this claim is correct, then for every $i$ we have $\trd_{i}(t^+) \le \trd_{i+1}(t^-) \le 2^{i+1}$, and the proof is complete.

    If $\alevel{\req}(t^-) = \level{\req}(t^-)$, the claim holds as levels do not decrease.
    Otherwise, we have that $\dist{\req}{a} > 2^{\alevel{\req}(t^-) -1} $.
    But then the triangle inequality implies that $\dist{\req}{a'} \ge \dist{\req}{a} - \dist{a}{a'} > 2^{\alevel{\req}(t^-) -1} - 2^{\ell+1} \ge 2^{\alevel{\req}(t^-)-2}$, where the final inequality uses $\alevel{\req}(t^-) \ge \ell+4$.
    This implies that $\alevel{\req}(t^+) \ge \alevel{\req}(t^-) - 1$, completing the proof.
\end{proof}

\begin{proposition}
    \label{prop:OSY_MovingDistanceBounds}
    During a service $\serv$, if the algorithm moves its server from $a$ to $a'$ in \cref{line:OSY_MoveServer}, then $2^{\level{\serv}-5} - 2^{\level{\serv}-8} \le \dist{a}{a'} \le 2^{\level{\serv}-3} + 2^{\level{\serv} - 8}$.
\end{proposition}
\begin{proof}
    First, we prove that $\dist{a}{a'} \ge 2^{\level{\serv}-5} - 2^{\level{\serv}-8} $.
    Since $\serv$ was started, we know that $\trd_{\level{\serv}-3} \ge 2^{\level{\serv} - 3}$.
    From \cref{prop:OSY_NoSuperCritical}, we know that $\trd_{\level{\serv}-5} \le 2^{\level{\serv} - 4}$; moreover, the service $\serv$ is primary, and thus $\creqs_{\serv}$ contains requests of level at most $\level{\serv} - 5$.
    Thus, the total residual delay of $\creqs_\serv$ incurred inside $\ball{a_{\serv}}{2^{\level{\serv}-5}}$ is at most $2^{\level{\serv}-4}$.
    Now note that if $\dist{a}{a'} < 2^{\level{\serv}-5} - 2^{\level{\serv}-8}$ then $\ball{a'}{2^{\level{\serv}-8}} \subseteq \ball{a}{2^{\level{\serv}-5}}$, which contradicts the definition of $a'$.

    Second, we prove that $\dist{a}{a'} \le 2^{\level{\serv}-3} + 2^{\level{\serv} - 8}$.
    Assuming otherwise that $\dist{a}{a'} > 2^{\level{\serv}-3} + 2^{\level{\serv} - 8}$, we have that $\ball{a'}{2^{\level{\serv} - 8}}$ and $\ball{a}{2^{\level{\serv} - 3}}$ are disjoint, in contradiction to the former containing much of the residual delay of $\serv$.
\end{proof}

We define the cost of a service $\serv$, denoted $\cost{\serv}$, to be the total movement of the algorithm's server in $\serv$ plus the total amount by which investment counters are raised in $\serv$.

\begin{proposition}
    $\alg \le \sum_{\serv \in \servs} \cost{\serv}$.
\end{proposition}
\begin{proof}
    Note that every request $\req$ is eventually served, and upon service $\ctr_{\req}$ is at least the delay cost of that request.
    Thus, the sum of counters upper-bounds delay costs, and the raising of every counter is counted towards $\cost{\serv}$ for some $\serv$.
    All movement costs in $\alg$ are also attributed to the cost of some service.
\end{proof}

\begin{proposition}
    The total cost of service $\serv$ is $O(1)\cdot 2^{\level{\serv}}$.
\end{proposition}
\begin{proof}
    From \cref{prop:OSY_NoSuperCritical}, it holds that $\trd_{\level{\serv}} \le 2^{\level{\serv} + 1}$; thus, $2^{\level{\serv}+1}$ bounds the cost of raising counters on \cref{line:OSY_ZeroResidualDelay}.

    In addition, the cost of traversing the $\pcstree$ solution on \cref{line:OSY_TraverseTree} and investing in counters in \cref{line:OSY_Invest} can be bounded using the following argument: in the previous iteration, the cost of the $\pcstree$ solution was less than $6\cdot 2^{\level{\serv}}$, from the condition of the loop, which implies that the cost of the optimal solution was less than $6\cdot 2^{\level{\serv}}$; however, delay rises continuously, and thus this optimal solution applies to the final iteration as well (at cost at most $6\cdot 2^{\level{\serv}}$).
    Since the approximation algorithm $\pcstree$ that we use is a 3-approximation, its cost of its output can be bounded by $18\cdot 2^{\level{\serv}}$. The penalty part of the solution is paid exactly in \cref{line:OSY_Invest}, while the served part of the solution is traversed in DFS (at double the cost).
    Thus, the total cost of \cref{line:OSY_TraverseTree} and \cref{line:OSY_Invest} is at most $36\cdot 2^{\level{\serv}}$.

    Finally, the cost of moving the server in \cref{line:OSY_MoveServer} (if it takes place) is at most $2^{\level{\serv} -3} + 2^{\level{\serv}-5}$.
\end{proof}

\begin{proposition}
    \label{lem:OSY_AlgorithmBoundedByPrimaryAndCertified}
    $\alg \le O(1) \cdot \pr{\sum_{\serv \in \pservs} 2^{\level{\serv}} + \sum_{\serv \in \cservs} 2^{\level{\serv}}}$.
\end{proposition}
\begin{proof}
    The proof is similar to that of \cref{lem:OSD_AlgorithmBoundedByPrimaryAndCertified} for the deadline case.
    Every non-primary service of level $\ell$ certifies another service of level at least $\ell -5$ (through \cref{prop:OSY_CertificationLevelDifference}).
    Since every service is certified at most once, it holds that $\sum_{\serv \in \servs \setminus \pservs} 2^{\level{\serv}} \le O(1)\cdot \sum_{\serv \in \cservs} 2^{\level{\serv}}$.
\end{proof}

\subsubsection*{Bounding Primary Services}
As in the argument for deadlines, we identify a subset of primary services which we need to bound.
In fact, for delay, we identify two such subsets.
Define $a^*_{\serv}$ to be the final location of the optimum's server at $\stime{\serv}$.
We define two disjoint subsets of the primary services $\pservs$:
\begin{enumerate}
    \item Services in which the algorithm's server ended at the starting location (i.e. \cref{line:OSY_MoveServer} did not run), denoted $\pservss$.
    \item Services in which the algorithm's server moved to some location $a'$ (i.e. \cref{line:OSY_MoveServer} ran) such that $\dist{a^*_{\serv}}{a'} \ge 2^{\level{\serv}-7}$, denoted $\pservsf$.
\end{enumerate}
As stated in \cref{prop:OSY_BoundingPrimaryByFarAndStationary}, to bound the cost of all primary services it is enough to bound the cost of these two subsets.

\begin{proposition}
    \label{prop:OSY_BoundingPrimaryByFarAndStationary}
    \[
        \sum_{\serv \in \pservs} 2^{\level{\serv}} \le O(1)\cdot \opt + O(1) \cdot \sum_{\serv \in \pservsf} 2^{\level{\serv}} + O(1) \cdot \sum_{\serv \in \pservss} 2^{\level{\serv}}
    \]
\end{proposition}
\begin{proof}
    The proof is similar to the proof of \cref{prop:OSD_BoundingPrimaryByFar}.
    We define the potential function $\phi(t) := 4 \dist{a(t)}{a^*(t)}$, where $a(t), a^*(t)$ at the locations of the algorithm's server and the optimum's server at $t$, respectively.
    Note that the potential function equals $0$ at the beginning of the input, and can only take on positive values.
    Note that in $\pservss$, the final and initial server locations are the same; we thus only consider services in $\pservs\setminus \pservss$ in the potential argument.
    For every service $\serv \in \pservs \setminus\pservss$, we define $a'_{\serv}$ to be the final location of the server in $\serv$ (the value of $a'$ in $\UponCritical$).
    Following the argument for \cref{prop:OSD_BoundingPrimaryByFar}, we have
    \begin{equation}
        \label{eq:OSY_PotentialArgument}
        \sum_{\serv \in \pservs \setminus \pservss} \dist{a'_{\serv}}{a_{\serv}} \le 4 \opt + \sum_{\serv \in \pservs\setminus \pservss} \pr{\dist{a'_{\serv}}{a_{\serv}} + \Delta_{\serv}}
    \end{equation}

    Now, consider a service $\serv\in\pservs \setminus (\pservsf \cup \pservss)$.
    For ease, define $a^*_\serv := a^*(\stime{\serv})$.
    After $\serv$ the algorithm moves its server to $a'_{\serv}$ such that $\dist{a'_{\serv}}{a^*_{\serv}} \le 2^{\level{\serv} - 7}$.
    In addition, \cref{prop:OSY_MovingDistanceBounds} implies that $\dist{a_{\serv}}{a'_{\serv}} \ge 2^{\level{\serv}-5} - 2^{\level{\serv}-8} = 7\cdot 2^{\level{\serv} -8}$.
    Therefore:
    \begin{align*}
        \Delta_{\serv} &\le 4(\dist{a'_{\serv}}{a^*_{\serv}} - \dist{a_{\serv}}{a^*_{\serv}}) \\
        &\le 4\cdot \pr{\dist{a'_{\serv}}{a^*_{\serv}} - \dist{a_{\serv}}{a'_{\serv}} + \dist{a'_{\serv}}{a^*_{\serv}}}\\
        &\le 4\cdot \pr{2^{\level{\serv} - 6} - \dist{a_{\serv}}{a'_{\serv}}} \\
        &\le 4\cdot (-\frac{3}{7} \cdot \dist{a_{\serv}}{a'_{\serv}}) \\
        &\le -\dist{a_{\serv}}{a'_{\serv}}
    \end{align*}
    where the second inequality is due to the triangle inequality.
    Therefore we have $\dist{a_{\serv}}{a'_{\serv}} + \Delta_{\serv} \le 0$ for every $\serv \in \pservs\setminus (\pservsf \cup \pservss)$.
    Moreover, note that for every $\serv \in \pservs\setminus \pservss$ we have $\Delta_{\serv} \le 4\dist{a_{\serv}}{a'_{\serv}}$.
    Finally, note that for a service $\serv \in \pservs \setminus \pservss$, \cref{prop:OSY_MovingDistanceBounds} implies that $\dist{a_{\serv}}{a'_{\serv}} \ge 2^{\level{\serv} - 5} - 2^{\level{\serv}-8} =7\cdot 2^{\level{\serv}-8}$, and thus $2^{\level{\serv}} \le \frac{256}{7} \cdot \dist{a_{\serv}}{a'_{\serv}}$.
    In addition, $\dist{a_{\serv}}{a'_{\serv}} \le 2^{\level{\serv}-3}  + 2^{\level{\serv}-8} = \frac{33}{256}\cdot 2^{\level{\serv}}$.
    Combining all observations, we get
    \begin{align*}
        \sum_{\serv \in \pservs} 2^{\level{\serv}} &\le \sum_{\serv \in \pservss} 2^{\level{\serv}} + \sum_{\serv \in \pservs \setminus \pservss}2^{\level{\serv}} \\
        &\le  \sum_{\serv \in \pservss} 2^{\level{\serv}} + \frac{256}{7} \cdot \sum_{\serv \in \pservs \setminus \pservss} \dist{a_{\serv}}{a'_{\serv}} \\
        &\le \sum_{\serv \in \pservss} 2^{\level{\serv}} + \frac{256}{7} \cdot \pr{4\opt + \sum_{\serv \in \pservs \setminus \pservss} (\dist{a_{\serv}}{a'_{\serv}} + \Delta_{\serv}) }\\
        &\le \sum_{\serv \in \pservss} 2^{\level{\serv}} + \frac{256}{7} \cdot \pr{4\opt + \sum_{\serv \in \pservsf} (\dist{a_{\serv}}{a'_{\serv}} + \Delta_{\serv}) }\\
        &\le \sum_{\serv \in \pservss} 2^{\level{\serv}} + \frac{256}{7} \cdot \pr{4\opt + 5\sum_{\serv \in \pservsf} \dist{a_{\serv}}{a'_{\serv}}}\\
        &\le \sum_{\serv \in \pservss} 2^{\level{\serv}} + \frac{256}{7} \cdot \pr{4\opt + \frac{165}{256}\sum_{\serv \in \pservsf} 2^{\level{\serv}}}\\
        &=O(1)\cdot \opt + O(1) \cdot \sum_{\serv \in \pservsf} 2^{\level{\serv}} + O(1) \cdot \sum_{\serv \in \pservss} 2^{\level{\serv}}
    \end{align*}
\end{proof}

\begin{definition}[primary interals and cylinders]
    For every service $\serv \in \pservsf \cup \pservss$, we define:
    \begin{enumerate}
        \item The primary time interval $\pivl{\serv} := \I{\min_{\req \in \creqs_{\serv}} \rlt{\req}}{\stime{\serv}}$.
        \item The primary cylinder $\pcyl{\serv} := \pr{\ball{a_\serv}{2^{\level{\serv}-2}}, \pivl{\serv}}$.
    \end{enumerate}
    We also define $\pcylsf, \pcylss$ to be the sets of all primary cylinders of services from $\pservsf, \pservss$, respectively.
    We define $\pcyls = \pcylsf \cup \pcylss$.
    In addition, for every $i$ we define $\pcyls_i$ to be the subset of primary cylinders from $\pcyls$ that belong to level-$i$ services.
\end{definition}

\begin{definition}[$\pdchg{\serv}$]
    For every primary service $\serv \in\pservs$, let $R \subseteq \creqs_{\serv}$ be the set of requests unserved in the optimal solution at $\stime{\serv}$.
    Define $\pdchg{\serv} := \sum_{\req \in R} \rylt{\req}{\stime{\serv}}$. (Here, recall that $\stime{\serv}$ refers to the time immediately before the service $\serv$.)
\end{definition}

\begin{proposition}
    \label{prop:OSY_DisjointImpliesChargePrimary}
    Let $\servs' \subseteq \pservsf \cup \pservss$ be a set of services such that their primary cylinders are disjoint.
    Then it holds that
    $\sum_{\serv \in \servs'} \pr{\ccap{\opt}{\pcyl{\serv}} + \pdchg{\serv}} \le \opt$.
\end{proposition}
\begin{proof}
    Denote by $\opt^m, \opt^d$ the movement and delay costs of the optimal solution, respectively.
    One can observe, as in \cref{obs:OSD_DisjointImpliesCharge}, that $\sum_{\serv \in \servs'} \ccap{\opt}{\ccyl{\serv}} \le \opt^m$.
    It remains to show that $\sum_{\serv \in \servs'} \pdchg{\serv} \le \opt^d$.

    Recall that the definition of $\pdchg{\serv}$ is $\sum_{\req \in R} \rylt{\req}{\stime{\serv}}$, where $R\subseteq \creqs_{\serv}$ is the subset of requests unserved by the optimal solution until $\stime{\serv}$.
    Note that the optimal solution indeed incurs $\rylt{\req}{\stime{\serv}} = \ylt{\req}{\stime{\serv}} - \ctr_{\req}{\stime{\serv}}$ delay for every request $\req \in R$.
    Moreover, delay is never charged twice, as the service raises $\ctr_{\req}$ to be $\ylt{\req}{\stime{\serv}}$ at service $\serv$.
    This completes the proof.
\end{proof}

\begin{proposition}
    \label{prop:OSY_PrimaryCylindersDisjoint}
    For every $i$, the set $\pcyls_i$ is a set of disjoint cylinders.
\end{proposition}
\begin{proof}
    Assume for contradiction that there exist two services $\serv_1,\serv_2$ of level $i$ such that $\pcyl{\serv_1}, \pcyl{\serv_2}$ are not disjoint.
    As the cylinders' time intervals are not disjoint, without loss of generality, assume that $\stime{\serv_1} \in \pivl{\serv_2}$.
    Since the cylinders are also not spatially disjoint, it holds that $\dist{a_{\serv_1}}{a_{\serv_2}} \le 2^{i-1}$.
    Now, from the definition of $\pivl{\serv_2}$, there exists a request $\req \in \creqs_{\serv_2}$ such that $\req$ is pending at $\stime{\serv_1}$.
    Now, note that since $\req \in \creqs_{\serv_2}$, it holds that $\dist{\req}{a_{\serv_2}} \le 2^{i-3}$, which implies $\dist{\req}{a_{\serv_1}} \le 2^{i-1} + 2^{i-3} \le 2^i$.
    Moreover, $\level{\req} \le i$ at $\stime{\serv_1}$.
    These facts imply that $\req$ was eligible for $\serv_1$, but this would imply that $\req$ increases in level to $i+1$ after $\serv_1$, in contradiction to $\req \in \creqs_{\serv_2}$.
\end{proof}

\begin{proposition}
    \label{prop:OSY_PrimaryCylinderIntersectionFar}
    For every $\serv \in \pservsf$, it holds that
    \[
        2^{\level{\serv}-8} \le \ccap{\opt}{\pcyl{\serv}} + \pdchg{\serv}.
    \]
\end{proposition}
\begin{proof}
    Consider the location $a'$ to which the algorithm moved its server at the end of $\serv$.
    From the definition of $\pservsf$, we know that $a^*(\stime{\serv}) \notin \ball{a'}{2^{\level{\serv}-7}}$.
    However, we also know from the definition of $a'$ that at least $2^{\level{\serv} -4}$ of the residual delay of $\serv$ accumulated inside the ball $\ball{a'}{2^{\level{\serv} - 8}}$.
    Thus, at least one of the following holds:
    \begin{enumerate}
        \item The optimal server visited $\ball{a'}{2^{\level{\serv} - 8}}$ during $\pivl{\serv}$ and left by $\stime{\serv}$; thus, it incurred a moving cost of at least $2^{\level{\serv} - 8}$ inside $\ball{a'}{2^{\level{\serv}-7}}$.
        Now, note that $\dist{a'}{a} + 2^{\level{\serv}-7} \le 2^{\level{\serv} - 3} + 2^{\level{\serv}-8} + 2^{\level{\serv}-7} < 2^{\level{\serv}-2}$, and thus $\ball{a'}{2^{\level{\serv}-7}} \subseteq \ball{a}{2^{\level{\serv}-2 cx}}$.
        Thus, $\ccap{\opt}{\pcyl{\serv}} \ge 2^{\level{\serv} - 8}$.
        \item The optimal server did not visit $\ball{a'}{2^{\level{\serv} - 8}}$ during $\pivl{\serv}$.
        In this case, it must be that $\pdchg{\serv} \ge 2^{\level{\serv}-4}$.
    \end{enumerate}
    In both cases, $2^{\level{\serv}-8} \le \ccap{\opt}{\pcyl{\serv}} + \pdchg{\serv} $.
\end{proof}

\begin{proposition}
    \label{prop:OSY_PrimaryCylinderIntersectionStationary}
    For every $\serv \in \pservss$, it holds that
    \[
        2^{\level{\serv}-8} \le \ccap{\opt}{\pcyl{\serv}} + \pdchg{\serv}.
    \]
\end{proposition}
\begin{proof}
    Let $a^*$ be the location of the optimal server at $\stime{\serv}$.
    At least one of the following options holds:
    \begin{enumerate}
        \item $a^* \notin \ball{a_{\serv}}{1.5\cdot 2^{\level{\serv}-3}}$.
        In this case, we have the following subcases.
        Either the optimum did not visit $\ball{a_{\serv}}{2^{\level{\serv}-3}}$ during $\pivl{\serv}$, in which case all requests in $\creqs_{\serv}$ remain unserved in $\opt$ at $\stime{\serv}$, and $\pdchg{\serv} \ge 2^{\level{\serv}-3}$; or, the optimal solution visited $\ball{a_{\serv}}{2^{\level{\serv}-3}}$ during $\pivl{\serv}$, which in turn implies that $\ccap{\opt}{\pcyl{\serv}} \ge 2^{\level{\serv}-4}$.
        In both cases, $ \ccap{\opt}{\pcyl{\serv}} + \pdchg{\serv} \ge 2^{\level{\serv} - 4}$.

        \item $a^* \in \ball{a_{\serv}}{1.5\cdot 2^{\level{\serv}-3}}$.
        In this case, consider $\ball{a^*}{2^{\level{\serv}-8}}$: since $\serv \in \pservss$, it must be that the residual delay of $\serv$ inside $\ball{a^*}{2^{\level{\serv}-8}}$ is at most $2^{\level{\serv}-4}$.
        But, the total residual delay of $\serv$ is at least $2^{\level{\serv}-3}$.
        Thus, at least $2^{\level{\serv}-4}$ residual delay of $\serv$ is outside $\ball{a^*}{2^{\level{\serv}-8}}$.
        If $\opt$ was outside $\ball{a^*}{2^{\level{\serv}-8}}$ during $\pivl{\serv}$, it must be that $\ccap{\opt}{\pcyl{\serv}} \ge 2^{\level{\serv}-8}$ (for this, note that $\ball{a^*}{2^{\level{\serv}-8}} \subseteq \ball{a_{\serv}}{2^{\level{\serv}-2}}$ and thus inside $\pcyl{\serv}$).
        Otherwise, $\opt$ remained in $\ball{a^*}{2^{\level{\serv}-8}}$ during $\pivl{\serv}$, and thus $\pdchg{\serv} \ge 2^{\level{\serv}-4}$.
        In both cases, $ \ccap{\opt}{\pcyl{\serv}} + \pdchg{\serv} \ge 2^{\level{\serv} - 8}$.
    \end{enumerate}
    In all cases, $\ccap{\opt}{\pcyl{\serv}} + \pdchg{\serv} \ge 2^{\level{\serv}-8}$.
\end{proof}

\begin{lemma}
    \label{lem:OSY_BoundingPrimary}
    $\sum_{\serv \in \pservsf \cup \pservss} 2^{\level{\serv}} \le O(\log \npoint) \cdot \opt$.
\end{lemma}
\begin{proof}
    Define $\rho = 2^{9}\cdot \npoint^2$.
    Combining \cref{prop:OSY_PrimaryCylinderIntersectionStationary,prop:OSY_PrimaryCylinderIntersectionFar,cor:OSD_PerforationCylinderIntersectionDifference}, for every service $\serv \in \pservsf \cup \pservss$, we have that
    \begin{align*}
        2^{\level{\serv}-8} &\le \ccap{\opt}{\cyl_{\serv}} + \pdchg{\serv} \\
        &\le \ccap{\percyl{\rho}_{\serv}}{\opt} + 2\cdot  2^{\level{\serv}-1} \cdot \npoint^2 /\rho + \pdchg{\serv} \\
        &\le \ccap{\percyl{\rho}_{\serv}}{\opt}  + 2^{\level{\serv} - 9} + \pdchg{\serv}
    \end{align*}
    Simplifying, we get $2^{\level{\serv} - 9} \le \ccap{\percyl{\rho}_{\serv}}{\opt} + \pdchg{\serv}$.
    Summing over all $\serv \in \pservsf \cup \pservss$, we have
    \[
        \sum_{\serv \in  \pservsf \cup \pservss} 2^{\level{\serv}} \le O(1)\cdot \sum_{\serv \in  \pservsf \cup \pservss} \pr{\ccap{\percyl{\rho}_{\serv}}{\opt} + \pdchg{\serv} }.
    \]
    Note that for every $i$, the cylinders of $\pcyls_i$ are disjoint (\cref{prop:OSY_PrimaryCylindersDisjoint}).
    Thus, defining $\percyls := \pc{\percyl{\rho}_{\serv} | \cyl_{\serv} \in \ccyls}$, \cref{prop:OSD_PerforationDisjointness} implies that $\percyls$ can be partitioned into $O(\log \npoint)$ disjoint sets.
    \Cref{prop:OSY_DisjointImpliesChargeCertified} thus implies that $\sum_{\serv \in \cservs} 2^{\level{\serv}} \le O(\log \npoint) \cdot \opt$, completing the proof.
\end{proof}

\subsubsection*{Bounding Certified Services}

In this subsection, we would like to prove the following.
\begin{lemma}
    \label{lem:OSY_BoundingCertified}
    $\sum_{\serv \in \cservs} 2^{\level{\serv}} \le O(\log \npoint) \cdot \opt$.
\end{lemma}

\begin{definition}
    Let $\serv \in \cservs$ be a certified service.
    Let $\serv' \in\cservs$ be the certified service with maximum $\ftime{\serv'}$ subject to $\level{\serv'} = \level{\serv}$, $\ftime{\serv'} \le \stime{\serv}$ and $\dist{\req_{\serv'}}{\req_{\serv}} < 6\cdot 2^{\level{\serv}}$.
    We define:
    \begin{enumerate}
        \item The time $\ptime{\serv} := \ftime{\serv'}$ if $\serv'$ exists (otherwise, define $\ptime{\serv} = -\infty$).
        \item The time interval $\civl{\serv} := \I{\ptime{\serv}}{\ftime{\serv}}$; note that $\stime{\serv}\in\civl{\serv}$.
    \end{enumerate}
\end{definition}

\begin{definition}[$\cpen{\serv}$ and $\cdchg{\serv}$]
    For every certified service $\serv \in\cservs$:
    \begin{enumerate}
        \item Define $\cpen{\serv}$ on requests $\req \in \elig_{\serv}$ such that
        \[
            \cpen{\serv}(\req) := \prp{\ylt{\req}{\ftime{\serv}} - \max\pc{\ylt{\req}{\stime{\serv}}, \ctr_{\req}(\stime{\serv})}}.
        \]
        (Here, recall that $\stime{\serv}$ refers to the time immediately before the service $\serv$.)
        \item Let $R \subseteq \elig_{\serv}$ be the subset of $\serv$-eligible requests unserved in the optimal solution at $\ftime{\serv}$.
        Define $\cdchg{\serv} := \sum_{\req \in R} \cpen{\serv}(\req)$.
    \end{enumerate}
\end{definition}

\begin{proposition}
    [analogue of \cref{prop:OSD_CertifierBetweenCertified}]
    \label{prop:OSY_CertifierBetweenCertified}
    Let services $\serv_1, \serv_2 \in \cservs$ be such that $\level{\serv_1} = \level{\serv_2} = \ell$ and $\dist{a_{\serv_1}}{a_{\serv_2}} \le 6\cdot 2^{\ell}$.
    Assuming WLOG that $\stime{\serv_1} < \stime{\serv_2}$, and letting $\serv$ be the service that made $\serv_1$ certified, it holds that $\stime{\serv}\in\I{\ftime{\serv_1}}{\stime{\serv_2}}$.
    (In particular, $\ftime{\serv_1} < \stime{\serv_2}$.)
\end{proposition}
\begin{proof}
    The two possible cases which contradict our proposition are that $\stime{\serv} \le \ftime{\serv_1}$ or that $\stime{\serv} >\stime{\serv_2}$.
    First, we prove that $\stime{\serv} > \ftime{\serv_1}$
    Consider the triggering request $\req \in \creqs_\serv$ that made $\serv$ certify $\serv_1$: this request was a witness for $\serv_1$, and thus in $\elig_{\serv_1}$; however, $\serv_1$ maintains that eligible requests will not accumulate residual delay until after time $\ftime{\serv_1}$.
    But, requests in $\creqs_{\serv}$ have positive residual delay at $\stime{\serv}$; thus $\stime{\serv} > \ftime{\serv_1}$.

    Now, assume for contradiction that $\stime{\serv} >\stime{\serv_2}$.
    Consider the service $\serv'$ which certified $\serv_2$; it must also be the case that $\stime{\serv'} > \stime{\serv_2}$.
    Suppose that $\stime{\serv} < \stime{\serv'}$.
    In this case, observe that all witnesses for $\serv_2$ at $\stime{\serv}$ are in $\ball{a_{\serv_2}}{2^{\ell}}$; but it holds through triangle inequality that
    \[
        \dist{a_{\serv_2}}{a_{\serv}} \le \dist{a_{\serv_2}}{a_{\serv_1}} +\dist{a_{\serv_1}}{\req} + \dist{\req}{a_{\serv}} \le   6\cdot 2^{\ell} + 2^{\ell} + 2^{\level{\serv} - 3}
    \]
    Using \cref{prop:OSY_CertificationLevelDifference}, $\level{\serv} \le \ell+5$, yielding that $\dist{a_{\serv_2}}{a_{\serv}} \le 11\cdot 2^{\ell}$.
    This implies that $\ball{a_{\serv_2}}{2^{\ell}} \subseteq \ball{a_\serv}{12\cdot 2^\ell}$.
    However, \cref{prop:OSY_CertificationLevelDifference} again implies $\level{\serv}\ge \ell+4$, and thus $\ball{a_{\serv_2}}{2^{\ell}} \subseteq \ball{a_{\serv}}{2^{\level{\serv}}}$.
    Combine this with the fact that the level of all witnesses for $\serv_2$ at $\stime{\serv}$ is at most $\ell + 1$ which is less than $\level{\serv}$; we thus obtain that all witnesses to $\serv_2$ at $\stime{\serv}$ are in $\elig_{\serv}$.
    But, after $\serv$, these witnesses would no longer be witnesses for $\serv_2$, in contradiction to one of them triggering $\serv'$ and certifying $\serv_2$.
    Similarly, if $\stime{\serv} >\stime{\serv'}$, the service $\serv'$ would leave no witnesses for $\serv_1$ to trigger $\serv$, which is a contradiction.
\end{proof}

\begin{definition}[certified cylinders]
    For a certified service $\serv \in \cservs$, define the certified cylinder $\ccyl{\serv} := (\ball{\req_{\serv}}{3\cdot 2^{\level{\req}}}, \civl{\serv})$.
    Define $\ccyls$ to be the set of all certified cylinders; in addition, for every $i$ define $\ccyls_i$ to be the set of certified cylinders formed from level-$i$ services.
\end{definition}

\begin{proposition}
    \label{prop:OSY_DisjointImpliesChargeCertified}
    Let $\servs' \subseteq \cservs$ be a set of services such that their certified cylinders are disjoint.
    Then it holds that
    \[
        \sum_{\serv \in \servs'} \pr{\ccap{\opt}{\ccyl{\serv}} + \cdchg{\serv}} \le \opt.
    \]
\end{proposition}
\begin{proof}
    Denote by $\opt^m, \opt^d$ the movement and delay costs of the optimal solution, respectively.
    One can observe, as in \cref{obs:OSD_DisjointImpliesCharge}, that $\sum_{\serv \in \servs'} \ccap{\opt}{\ccyl{\serv}} \le \opt^m$.
    It remains to show that $\sum_{\serv \in \servs'} \cdchg{\serv} \le \opt^d$.

    Recall that the definition of $\cdchg{\serv}$ is $\sum_{\req \in R} \cpen{\serv}(\req)$, where $R\subseteq \elig_{\serv}$ is the set of requests unserved by the optimal solution until $\ftime{\serv}$.
    For every $\req \in R$, note that $\ctr'_{\req} := \max\pc{\ylt{\req}{\stime{\serv}}, \ctr_{\req}(\stime{\serv})}$ is exactly the value of the counter $\ctr_{\req}$ after \cref{line:OSY_ZeroResidualDelay} of $\serv$; thus, $\cpen{\serv}(\req) =\prp{ \ylt{\req}{\stime{\serv}} - \ctr'_{\req}} $.
    Let $t' \in [\stime{\serv}, \ftime{\serv}]$ be the point in time in which $\ylt{\req}{t'} = \ctr'_{\req}$; the optimal solution incurs a delay cost of $\ylt{\req}{\ftime{\serv}} - \ctr'_{\req}$ during the interval $[t', \ftime{\serv}]$, and thus the charging $\cdchg{\serv}$ is valid.
    Moreover, $\serv$ either serves $\req$, or raises $\ctr_{\req}$ to at least $\ylt{\req}{\ftime{\serv}}$ (in \cref{line:OSY_Invest}); thus, the delay of $\req$ during $[t', \ftime{\serv}]$ is only charged once to the optimal solution.
    This completes the proof.
\end{proof}

\begin{proposition}
    \label{prop:OSY_CertifiedCylindersDisjoint}
    For every $i$, $\ccyls_i$ is a set of disjoint cylinders.
\end{proposition}
\begin{proof}
    The proposition results from \cref{prop:OSY_CertifierBetweenCertified} in the same way that \cref{prop:OSD_CertifiedCylindersDisjoint} results from \cref{prop:OSD_CertifierBetweenCertified}.
\end{proof}

\begin{proposition}
    \label{prop:OSY_RequestInCertifiedCylinder}
    For every certified service $\serv \in \cservs$ and request $\req \in \elig_{\serv}$, it holds that $\rlt{\req} > \ptime{\serv}$.
\end{proposition}
\begin{proof}
    If $\ptime{\serv} = -\infty$ then we are done.
    Otherwise, there exists a certified service $\serv' \in \cservs$ such that $\level{\serv'} = \level{\serv}$, $\ftime{\serv'} = \ptime{\serv}$, and $\dist{a_\serv}{a_\serv'} \le 6\cdot 2^{\level{\serv}}$.
    Define $\ell := \level{\serv}$.
    Using \cref{prop:OSY_CertifierBetweenCertified}, the service $\serv''$ that certified $\serv'$ occured in the interval $\I{\ftime{\serv'}}{\stime{\serv}}$.
    Thus, there must exist a witness request that made $\serv''$ certify $\serv'$; thus, there exists a request in $\elig_{\serv'} \cap \creqs_{\serv''}$.
    But this means that $\dist{a_{\serv'}}{a_{\serv''}} \le 2^{\level{\serv'}} + 2^{\level{\serv''}-3} \le 5\cdot 2^{\ell}$, where the second inequality uses \cref{prop:OSY_CertificationLevelDifference} for $\level{\serv''} \le \level{\serv}+5$.

    Assume for contradiction that $\rlt{\req} \le \ptime{\serv} = \ftime{\serv'}$.
    Thus, $\req$ is pending at $\stime{\serv''}$, and also has level at most $\ell$ (since $\req \in \elig_{\serv}$).
    In addition we have:
    \begin{align*}
        \dist{\req}{a_{\serv''}} &\le \dist{\req}{a_{\serv}} + \dist{a_{\serv}}{a_{\serv'}} + \dist{a_{\serv'}}{a_{\serv''}} \\
        &\le 2^{\ell} + 6\cdot 2^{\ell} + 5\cdot 2^{\ell} = 12\cdot 2^{\ell} \le 2^{\level{\serv''}}
    \end{align*}
    where the final inequality uses \cref{prop:OSY_CertificationLevelDifference} to claim that $\level{\serv''} \ge \level{\serv} + 4 $.
    Thus, $\req \in \elig_{\serv''}$, and since it is not served by $\serv''$, its level after $\serv''$ is at least $\level{\serv''} + 1 \ge \ell + 5$; but this contradicts the level of $\req$ being at most $\ell$ at $\stime{\serv}$.
    This completes the proof of the proposition.
\end{proof}

\begin{proposition}
    \label{prop:OSY_CertifiedCylinderIntersection}
    For every certified service $\serv \in \cservs$, it holds that $ 2^{\level{\serv}} \le  2\ccap{\opt}{\ccyl{\serv}} + \cdchg{\serv}$.
\end{proposition}
\begin{proof}
    Denote by $R \subseteq \elig_{\serv}$ the set of requests in $\elig_{\serv}$ whose location was visited by during $\civl{\serv}$.
    Thus, using \cref{prop:OSD_BallIntersection}, it holds that $\stree^*(R) \le 2\ccap{\opt}{3\cdot 2^{\level{\serv}}}$.
    Now, note that $\stree^*(R\cup\pc{a_{\serv}}) \le \stree^*(R) + 2^{\level{\serv}}$.
    Finally, note that the prize-collecting Steiner tree problem whose solution is traversed by the algorithm uses the penalty function $\cpen{\serv}$.
    Combining, we get
    \begin{align*}
        2^{\level{\serv} + 1} &\le \pcstree^*(\elig_{\serv}, \pi_{\serv}; a_{\serv}) \\
        &\le  \stree^*(R\cup\pc{a_{\serv}}) + \sum_{\req \in \elig_{\serv}\backslash R} \cpen{\serv}(\req) \\
        &\le \stree^*(R) + 2^{\level{\serv}} + \cdchg{\serv} \\
        &\le 2\ccap{\opt}{3\cdot 2^{\level{\serv}}}  + 2^{\level{\serv}} + \cdchg{\serv}
    \end{align*}
    Where the first inequality stems from \cref{line:OSY_DefineForwardingTime} together with the fact that $\pcstree$ is a 3-approximation and the second inequality is since serving only $R$ is a feasible solution to $\pcstree^*(\elig_{\serv}, \pi_{\serv}; a_{\serv})$.
    Simplifying, we get that $2^{\level{\serv}} \le 2\ccap{\opt}{\ccyl{\serv}} + \cdchg{\serv}$.
\end{proof}

\begin{proof}
    [Proof of \cref{lem:OSY_BoundingCertified}]
    Define $\rho = 24\cdot \npoint^2$.
    Combining \cref{prop:OSY_CertifiedCylinderIntersection} and \cref{cor:OSD_PerforationCylinderIntersectionDifference}, for every cylinder $\cyl_{\serv} \in \ccyls$, we have
    \begin{align*}
        2^{\level{\serv}} &\le 2\ccap{\opt}{\cyl_{\serv}} + \cdchg{\serv} \\
        &\le 2\ccap{\percyl{\rho}_{\serv}}{\opt} + 4\cdot 3\cdot 2^{\level{\serv}} \cdot \npoint^2 /\rho + \cdchg{\serv} \\
        &\le 2\ccap{\percyl{\rho}_{\serv}}{\opt}  + 2^{\level{\serv} - 1} + \cdchg{\serv}
    \end{align*}
    Thus, $\sum_{\serv \in \cservs} 2^{\level{\serv}} \le O(1)\cdot \sum_{\serv \in \servs} (\ccap{\opt}{\cyl_{\serv}} + \cdchg{\serv})$.
    Note that for every $i$, the cylinders of $\ccyls_i$ are disjoint (\cref{prop:OSY_CertifiedCylindersDisjoint}).
    Thus, defining $\percyls := \pc{\percyl{\rho}_{\serv} | \cyl_{\serv} \in \ccyls}$, \cref{prop:OSD_PerforationDisjointness} implies that $\percyls$ can be partitioned into $O(\log \npoint)$ disjoint sets.
    \Cref{prop:OSY_DisjointImpliesChargeCertified} thus implies that $\sum_{\serv \in \cservs} 2^{\level{\serv}} \le O(\log \npoint) \cdot \opt$, completing the proof.
\end{proof}

\begin{proof}
    [Proof of \cref{thm:OSY_Competitiveness}]
    The following holds:
    \begin{align*}
        \alg &\le O(1) \cdot \pr{\sum_{\serv \in \cservs} 2^{\level{\serv}} +\sum_{\serv \in \pservsf} 2^{\level{\serv}} +\sum_{\serv \in \pservss} 2^{\level{\serv}} + \opt  } \\
        &\le O(\log \npoint) \cdot \opt
    \end{align*}
    where the first inequality uses \cref{lem:OSY_AlgorithmBoundedByPrimaryAndCertified,prop:OSY_BoundingPrimaryByFarAndStationary}, and the second inequality uses \cref{lem:OSY_BoundingPrimary,lem:OSY_BoundingCertified}
\end{proof}

    \section{Extension to Request Regime}
    \label{sec:RR}
    In this section, we extend the results of the paper to obtain competitiveness as a function of $\nreqs$, the number of requests in the input.
Specifically, we prove the following theorem.

\begin{theorem}
    \label{thm:RR_OSYCompetitiveness}
    There exists a $O(\log \min\pc{\npoint,\nreqs})$-competitive algorithm for online service with delay (in particular, also for online service with deadlines).
\end{theorem}

Denote by  $\npoint'$ the number of points on which a request is released in the input, and note that $\npoint' \le \min\pc{\npoint,\nreqs}$; we show $O(\log \npoint')$-competitiveness for online service with deadlines or delay, thus proving \cref{thm:RR_OSYCompetitiveness}.

\subsection{Proof of Theorem \ref{thm:RR_OSYCompetitiveness}}

The algorithm and its proof are almost identical to that of \cref{sec:OSY}.
Thus, instead of reiterating the proof, we go over the necessary changes.

\textbf{Modified Algorithm.}
For every time $t$, define $\ms_t$ to be the metric space induced by requests that have been released by time $t$.
Formally, we observe the points in $\ms$ on which requests have been released, plus the initial location of the server; these are the vertices of the graph $\ms_t$.
Between every two vertices $u,v$ in $\ms_t$, there exists an edge of weight $\dist{u}{v}$.
Note that $\ms_{t_1}$ is a subgraph of $\ms_{t_2}$ for every $t_2 \ge t_1$.
Define $\ms' := \ms_{\infty}$, and note that the number of vertices in $\ms'$ is at most $\npoint'+1$.

We modify \cref{alg:OSY} to consider only points on which requests have previously been released.
Specifically, we make the following changes to $\UponCritical$ called at time $t$:
\begin{enumerate}
    \item Whenever the algorithm identifies a new location for the server (\cref{line:OSY_FindNewLoc}), we only consider points in $\ms_{t}$.
    Hence, at time $t$, the server of the algorithm is always at a point in $\ms_{t}$.
    \item Whenever $\pcstree$ is called, we instead call $\pcstree_{\ms_{t}}$, which is the approximation for prize-collecting Steiner tree run on the graph $\ms_{t}$.
    Note that whenever we call $\pcstree$, the requested terminals and the root node all belong to $\ms_{t}$, and thus this is well-defined.
\end{enumerate}

In addition, we can assume without loss of generality that the optimal solution is lazy, and only moves its server to the location of a pending request.
Thus, at any time $t$ the optimal solution's server exists only in $\ms_{t}$.
At this point, the proofs of the lemmas and propositions of \cref{sec:OSY} go through, where the underlying graph is $\ms'$.
This yields a competitive ratio of $O(\log \npoint')$, which is at most $O(\log \min\pc{\npoint, \nreqs})$.

    \section{Improved Analysis through Perforated Cylinders}
    \label{sec:Perf}
    \textbf{Perforated balls and cylinders.}

\begin{definition}[perforated balls and cylinders]
    For every point $v \in \ms$, radius $r$, and number $\rho > 1$, define the perforated ball
    \[
        \perball{\rho}{v}{r} := \ball{v}{r} - \bigcup_{v'\in\ms} \ball{v'}{\frac{r}{\rho}}.
    \]
    In addition, given a cylinder $\cyl = (\ball{v}{r}, I)$, define the perforated cylinder $\percyl{\rho} := (\perball{\rho}{v}{r}, I)$.
\end{definition}

\begin{proposition}
    \label{prop:OSD_PerforationIntersectionDifference}
    For every subgraph $G'$, and every choice of $v,r,\rho$, it holds that $\ccap{G'}{\ball{v}{r}} \le \ccap{G'}{\perball{\rho}{v}{r}} + \frac{2r\npoint^2}{\rho}$.
\end{proposition}
\begin{proof}
    Consider every edge $e$ in $G'$.
    The total weight of the edge $e$ contained in balls of radius $r/\rho$ is at most $2r/\rho$ (specifically, this weight is contained in the intersection with the balls centered in the two endpoints of the edge).
    The fact that the number of edges in $G'$ is at most $\npoint ^2$ completes the proof.
\end{proof}
The following corollary follows immediately.

\begin{corollary}
    \label{cor:OSD_PerforationCylinderIntersectionDifference}
    For every cylinder $\cyl = (\ball{v}{r}, I)$ and parameter $\rho$, it holds that
    \[
        \ccap{\opt}{\cyl} \le \ccap{\opt}{\percyl{\rho}} + \frac{2r\npoint^2}{\rho}
    \]
\end{corollary}

\begin{proposition}
    \label{prop:OSD_PerforationDisjointness}
    Suppose that for every integer $i$, $\cyls_i$ is some set of disjoint cylinders of the form $(\ball{v}{x}, I)$ where $\ceil{\log x} = i$, and define $\cyls = \bigcup_i \cyls_{i}$.
    Then for every parameter $\rho \ge 2$, and defining $\percyls := \pc{\percyl{\rho} | \cyl \in \cyls}$, the set $\percyls$ can be partitioned into $O(\log \rho)$ sets of disjoint cylinders.
\end{proposition}
\begin{proof}
    Define $\levsep = \ceil{\log \rho}+1$, and define $\percyls_i = \pc{\percyl{\rho} | \cyl \in \cyls_i}$.
    For every $i \in [\levsep]$, define $\overline{\percyls}_i = \bigcup_{j \in i + \levsep\mathbb{Z}} \percyls_j$; note that $\percyls$ is partitioned into those $\levsep$ sets.

    It remains to show that $\overline{\percyls}_i$ is a set of disjoint cylinders for every $i$.
    Consider two cylinders $\percyl{\rho}_1, \percyl{\rho}_2 \in \overline{\percyls}_i$, denote the centers of their balls by $v_1, v_2$, and denote the radii of their balls by $r_1, r_2$, respectively.
    If $\ceil {\log r_1}= \ceil {\log r_2} = i$, then $\percyl{\rho}_1,\percyl{\rho}_2 \in \percyls_i$; in this case, they must be disjoint: $\cyls_i$ is a set of disjoint cylinders, and $\percyl{\rho} \subseteq \cyl$ for every $\cyl$.
    Otherwise, $r_1\neq r_2$; assume WLOG that $\ceil {\log r_1} \ge \ceil {\log r_2} + \levsep$.
    Now, consider that $\perball{\rho}{v_1}{r_1}$ and $\ball{v_2}{r_1/\rho}$ are disjoint, by construction; now note that $r_2 \le r_1 / 2^{\levsep-1} \le r_1 / \rho$, and thus $\perball{\rho}{v_2}{r_2} \subseteq \ball{v_2}{r_1/\rho}$, which shows that $\percyl{\rho}_1,\percyl{\rho}_2$ are disjoint.
    Overall, we proved that $\overline{\percyls}_i$ is a disjoint set.
\end{proof}

\begin{proof}[Proof of \cref{lem:OSD_BoundingFarPrimaryImproved}]
    Define $\rho = 2^{6}\npoint^2$.
    Combining \cref{prop:OSD_PrimaryCylinderIntersection} and \cref{cor:OSD_PerforationCylinderIntersectionDifference}, for every cylinder $\cyl_{\serv} \in \pcyls$, we have
    \begin{align*}
        2^{\level{\serv}-6} &\le \ccap{\opt}{\cyl_{\serv}} \\
        &\le \ccap{\percyl{\rho}_{\serv}}{\opt} + 2\cdot 2^{\level{\serv}-2} \cdot \npoint^2 /\rho \\
        &\le \ccap{\percyl{\rho}_{\serv}}{\opt}  + 2^{\level{\serv} - 7}.
    \end{align*}
    Simplifying, we get $2^{\level{\serv}-7} \le \ccap{\percyl{\rho}_{\serv}}{\opt}$.
    Defining the set of cylinders $\ppercyls := \cset{\percyl{\rho}_{\serv}}{\cyl_{\serv} \in \pcyls}$, we have
    $\sum_{\serv \in \pservsf} 2^{\level{\serv}} \le O(1)\cdot \sum_{\serv \in \pservsf} \ccap{\opt}{\percyl{\rho}_{\serv}}$.
    Now, using \cref{prop:OSD_PerforationDisjointness}, we know that $\ppercyls$ can be partitioned into $O(\log \rho) = O(\log \npoint)$ sets of disjoint cylinders.
    Thus, $\sum_{\serv \in \pservsf} \ccap{\opt}{\percyl{\rho}_{\serv}} \le O(\log \npoint) \cdot \opt$.
\end{proof}

\begin{proof}[Proof of \cref{lem:OSD_BoundingCertifiedImproved}]
    Define $\rho = 24\cdot \npoint^2$.
    Combining \cref{prop:OSD_CertifiedCylinderIntersection} and \cref{cor:OSD_PerforationCylinderIntersectionDifference}, for every cylinder $\cyl_{\serv} \in \ccyls$, we have
    \begin{align*}
        2^{\level{\serv}-1} &\le \ccap{\opt}{\cyl_{\serv}} \\
        &\le \ccap{\percyl{\rho}_{\serv}}{\opt} + 2\cdot 3\cdot 2^{\level{\serv}} \cdot \npoint^2 /\rho \\
        &\le \ccap{\percyl{\rho}_{\serv}}{\opt}  + 2^{\level{\serv} - 2}.
    \end{align*}
    Simplifying, $2^{\level{\serv}-2} \le \ccap{\percyl{\rho}_{\serv}}{\opt}$.
    Defining $\percyls := \cset{\percyl{\rho}_{\serv}}{\cyl_{\serv} \in \ccyls}$, we have
    $\sum_{\serv \in \cservs} 2^{\level{\serv}} \le O(1)\cdot \sum_{\serv \in \cservs} \ccap{\opt}{\percyl{\rho}_{\serv}}$.
    \Cref{prop:OSD_PerforationDisjointness} yields that $\percyls$ can be partitioned into $O(\log \rho) = O(\log \npoint)$ sets of disjoint cylinders.
    Thus, $\sum_{\serv \in \cservs} \ccap{\opt}{\percyl{\rho}_{\serv}} \le O(\log \npoint) \cdot \opt$.
\end{proof}


    \bibliographystyle{src/ACM-Reference-Format}
    \bibliography{src/bibfile}

\end{document}